\documentclass[proc]{edpsmath}

\usepackage{booktabs}
\usepackage{graphicx}
\usepackage{hyperref}
\hypersetup{colorlinks = true, linkcolor = blue, citecolor =  red, urlcolor= black}

\makeatletter
\theoremstyle{remark}
\newtheorem*{note*}{\protect\notename}
\theoremstyle{plain}
\newtheorem*{assumption*}{\protect\assumptionname}
\theoremstyle{plain}
\newtheorem{thm}{\protect\theoremname}[section]
\theoremstyle{plain}
\newtheorem{lem}[thm]{\protect\lemmaname}

\allowdisplaybreaks

\newtheorem{proposition}{Proposition}[section]
\newtheorem{remark}{Remark}[section]

\newtheorem{corollary}{Corollary}[section]

 \makeatother

\usepackage{babel}
\providecommand{\assumptionname}{Assumption}
\providecommand{\lemmaname}{Lemma}
\providecommand{\notename}{Note}
\providecommand{\theoremname}{Theorem}

\newcommand{\E}{\mathbb{E}}
\newcommand{\esp}{\mathbb{E}}
\newcommand \Esp[1]{\E\left[#1\right]}

\renewcommand{\P}{\mathbb{P}}
\newcommand{\bH}{\mathbb{H}}
\newcommand{\bS}{\mathbb{S}}
\newcommand{\R}{\mathbb{R}}

\newcommand{\1}{\mathbf{1}}

\newcommand{\cF}{\mathcal{F}}
\newcommand{\cN}{\mathcal{N}}

\newcommand{\dd}{{\rm d}}
\newcommand{\ds}{{\rm d}s}
\newcommand{\dt}{{\rm d}t}
\newcommand{\du}{{\rm d}u}
\newcommand{\dS}{{\rm d}S}
\newcommand{\dV}{{\rm d}V}
\newcommand{\dW}{{\rm d}W}
\newcommand{\CVaRa}{\mbox{ \bf{CVaR}}^\alpha}
\newcommand{\CVaRaFt}{\mbox{ \bf{CVaR}}^\alpha_{\cF_t}}
\newcommand{\CVaRaFs}{\mbox{ \bf{CVaR}}^\alpha_{\cF_s}}

\numberwithin{equation}{section}

\newcommand{\eqlnostar}[2]{\begin{align}\label{#1}#2\end{align}}
\newcommand{\eqstar}[1]{\begin{align*}#1\end{align*}}
\newcommand{\eq}[1]{\ifthenelse{\equal{#1}{*}}
  {\eqstar}
  {\eqlnostar{#1}}
 }
\newcommand{\red}[1]{{\color{black}#1}}

\newcommand{\dC}{d^{\mathrm{Call}}}
\newcommand{\dP}{d^{\mathrm{Put}}}
\newcommand{\sigC}{\sigma^{\mathrm{Call}}}
\newcommand{\sigP}{\sigma^{\mathrm{Put}}}
\newcommand{\maj}{>}
\newcommand{\mino}{<}
\newcommand{\Prob}{\mathbb P}
\newcommand{\til}{~}

\makeatletter
\def\timenow{\@tempcnta\time
\@tempcntb\@tempcnta
\divide\@tempcntb60
\ifnum10>\@tempcntb0\fi\number\@tempcntb
:\multiply\@tempcntb60
\advance\@tempcnta-\@tempcntb
\ifnum10>\@tempcnta0\fi\number\@tempcnta}
\makeatother

\begin{document}
\title{Numerical approximations of McKean Anticipative Backward Stochastic Differential Equations\\ arising in Initial Margin requirements}
\thanks{This work has been carried out during the 6-weeks summerschool CEMRACS, July-August 2017, with the support of the Chaire Risques Financiers (Ecole Nationale des Ponts et Chauss\'ees, Ecole Polytechnique, Soci\'et\'e G\'en\'erale, Sorbonne Universit\'e - with the partnership of the Fondation du Risque).}
\author{A. Agarwal}\address{Adam Smith Business School, University of Glasgow, University Avenue, G12 8QQ Glasgow, United Kingdom. Email: {\tt  ankush.agarwal@glasgow.ac.uk}. This author research was conducted while at CMAP, Ecole Polytechnique, and is  part of the {\it Chaire Risque Financiers}.}
\author{S. De Marco}\address{Centre de Math\'ematiques Appliqu\'ees (CMAP), Ecole Polytechnique, Route de Saclay, 91128 Palaiseau Cedex, France. Email: {\tt  stefano.de-marco@polytechnique.edu}}
\author{E. Gobet}\address{CMAP, Ecole Polytechnique. Email: {\tt emmanuel.gobet@polytechnique.edu}. This author research is part of the {\it Chaire Risque Financiers} of the {\it Fondation du Risque} and of the {\it  Finance for Energy Market Research Centre}.}
\author{J. G. Lopez-Salas}\address{
Department of Mathematics, Faculty of Informatics, Universidade da Coru\~na, Campus de Elvi\~na s/n, 15071 - A Coru\~na, Spain. Email: {\tt jose.lsalas@udc.es}.
This author research was conducted while at CMAP, Ecole Polytechnique, and is  part of the {\it Chaire Risque Financiers}.}
\author{F. Noubiagain}\address{\red{This author research was conducted while at Laboratoire Manceau de Math\'{e}matiques, Le Mans Universit\'{e}, France. Email: {\tt  larissafanny@yahoo.fr}} }
\author{A. Zhou}\address{\red{This author research was conducted while at Universit\'e Paris-Est, CERMICS (ENPC), 77455, Marne-la-Vall\'ee, France. Email: {\tt alexandre.zhou@enpc.fr}}}

\begin{abstract}
We introduce a new class of anticipative backward stochastic differential equations with a dependence of McKean type on the law of the solution, that we name MKABSDE. We provide existence and uniqueness results in a general framework with relatively general regularity assumptions on the coefficients. We show how such stochastic equations arise within the modern paradigm of derivative pricing where a central counterparty (CCP) requires the members to deposit  variation and initial margins to cover their exposure. In the case when the initial margin is proportional to the Conditional Value-at-Risk (CVaR) of the contract price, we apply our general result to define the price as a solution of a MKABSDE. We provide several linear and non-linear simpler approximations, which we solve using different numerical (deterministic and Monte-Carlo) methods. 


\end{abstract}

\begin{resume}
Nous introduisons une nouvelle famille d'\'equations diff\'erentielles stochastiques r\'etrogrades anticipatives ayant une d\'ependance par rapport \`a la loi de la solution, que nous appelons MKABSDE. Ces \'equations apparaissent dans le contexte moderne de la valorisation de d\'eriv\'es en pr\'esence d'appels de marge de la part d'une chambre de compensation. 
Nous d\'emontrons un r\'esultat d'existence et unicit\'e sous des hypoth\`eses relativement faibles sur les coefficients de l'\'equation. Dans le cas o\`u les appels de marge sont proportionnels \`a la VaR conditionnelle (CVaR) du prix du contrat, notre r\'esultat g\'en\'eral entra\^ine l'existence et unicit\'e pour le prix en tant que solution d'une MKABSDE. Nous consid\'erons plusieurs approximations lin\'eaires et non-lin\'eaires de cette \'equation, que nous abordons avec diff\'erentes m\'ethodes num\'eriques.
\medskip

\centerline{\Large First version: January 29, 2018.
\\
This version: \today}
\end{resume}

\maketitle

\section{Initial margin and McKean Anticipative BSDE (MKABSDE)}
\subsection{Financial context and motivation}
The paradigm of linear risk-neutral pricing of financial contracts
has changed in the last few years, influenced by the regulators.
Nowadays, for several type of trades banking institutions have to post collateral to a central counterparty (CCP, also called clearing house) in order to secure their positions.
On a daily basis, the CCP imposes to each member to post a certain amount computed according to the estimated exposure of their contracts.
The variation margin deposit corresponds to a collateral that compensates the daily fluctuations of the market value of a contract, while the initial margin deposit is intended to reduce the gap risk, which is the possible mark-to-market loss encountered during the liquidation period upon default of the contract's counterpart (see \cite[Chapter 3]{CrepeyBielecki_book2014}, or the Basel Committee on Banking Supervision document \cite{base:15}, for more details). 
In this work we focus only on the  initial margin requirement (IM for short), and we investigate how it affects the valuation and hedging of the contract.
As stated in \cite[p.11 3(d)]{base:15}, \emph{``IM protects the transacting parties from the potential future exposure that could arise from future changes in the mark-to-market value of the contract during the time it takes to close out and replace the position in the event that one or more counterparties default. The amount of initial margin reflects the size of the potential future exposure. It depends on a variety of factors,  [\dots] the expected duration of the contract closeout and replacement period, and can change over time."} 
In this work, we will consider IM deposits that are proportional to the Conditional Value-at-Risk (CVaR) of the contract value over a future period of length $\Delta$ (typically $\Delta=1$ week or $10$ days, standing for the closeout and replacement period).
We focus on CVaR rather than Value-at-Risk (VaR) due to its pertinent mathematical properties; it is indeed well established that CVaR is a coherent risk measure whereas VaR is not \cite{artz:delb:eber:heat:99}.

We make some distinctions in our analysis according to the way the \emph{contract value} is computed in the presence of IM.
While \cite{base:15} refers to a mark-to-market value of the contract that can be seen as an \emph{exogenous} value, we investigate the case where this value is \emph{endogenous} and is given by the value of the hedging portfolio including the additional IM costs.
\red{A similar setting is considered in Nie and Rutkowski \cite{Nie2016}.}
By doing so, we introduce a new non-linear pricing rule, that is: the value of the hedging portfolio $V_t$ together with its hedging component $\pi_t$ solve a stochastic equation including a term depending on the law of the solution (due to the CVaR). We
justify that this problem can be seen as a new type of anticipative
Backward Stochastic Differential Equation (BSDE) with McKean interaction \cite{mcke:66}. From now on, we refer to this kind of equation as MKABSDE, standing for McKean Anticipative BSDE; Section \ref{subsection:A first Anticipative BSDE with dependence in law} below gives a toy example of such a model. We derive stability estimates
for these MKABSDEs, under general Lipschitz conditions, and prove existence and uniqueness results.
In Section 
\ref{section:The Case of CVaR variation margin}, we verify that these results can be applied to a general complete It\^o market \cite{kara:shre:98}, when accounting for IM requirements. Then, we derive some approximations based on classical non-linear BSDEs whose purpose is to quantify the impact of choosing the reference price for the IM as exogenous or endogenous, and to compare with the case without IM.
\red{Essentially, in Theorem \ref{theorem: approximation:price} we prove that the hedging portfolio with exogenous or endogenous reference price for the IM coincide up to order 1 in $\Delta$ when $\Delta$ is small (which is compatible with $\Delta$ equal to a few days), while the difference between valuation \emph{with} or \emph{without} IM correction has a size of order $\sqrt \Delta$.
Finally, Section \ref{section:Numerical Examples} is devoted to numerical experiments: we solve the different approximating BSDEs using finite difference methods in dimension 1, and nested Monte Carlo and regression Monte Carlo methods in higher dimensions.}

\subsection{An example of anticipative BSDE with dependence in law}
\label{subsection:A first Anticipative BSDE with dependence in law}

We start with a simple financial example with IM requirements, in the case of a single tradable asset.
A more general version with a multidimensional It\^o market will be studied in Section \ref{section:The Case of CVaR variation margin}.
Let us assume that the price of a tradable asset, denoted $S$, evolves accordingly to a geometric Brownian motion
\begin{equation}
\dS_{t}=\mu S_{t}\dt+\sigma S_{t}\dW_{t},\label{eq:SDE}
\end{equation}
where $(\mu,\sigma)\in\R\times \R^+$ and $W$ is an one-dimensional Brownian motion. 

In the classical financial setting (see, for example, \cite{musi:rutk:05}), consider the situation where a trader wants to
sell a European option with maturity $T>0$ and payoff $\Phi\left(S_{T}\right)$,
and to hedge it dynamically with risky and riskless assets $S$ and $S^{0}$,
where $S_{t}^{0}=e^{rt}$ for $t\in[0,T]$ and $r$ is a risk-free interest rate.
We denote by $\left(V,\pi\right),$ the value of the self-financing portfolio and $\pi$
the amount of money invested in the risky asset, respectively. In order to ensure the replication of the payoff at maturity, the couple $\left(V,\pi\right)$
should solve the following stochastic equation
\begin{equation}\label{eq:0}
\begin{cases}
\dV_{t} & = \displaystyle r\left(V_{t}-\pi_{t}\right)\dt+\pi_t\frac{\dS_{t}}{S_{t}},\ t\in[0,T],\\
V_{T} & =  \Phi(S_{T}).
\end{cases}
\end{equation}
Eq \eqref{eq:0} is a BSDE since the terminal condition of $V$ is imposed.
Because all the coefficients are linear in $V$ and $\pi$, \eqref{eq:0} is a linear BSDE (see \cite{kpq1997} for a broad overview on BSDEs and their applications in finance). 

Accounting for IM requirement will introduce an additional cost in the above self-financing dynamics.
We assume that the required deposit is proportional to the CVaR of the portfolio
over $\Delta$ days (typically $\Delta = $ 10 days) at the risk-level $\alpha$ (typically $\alpha=99\%$).
The funding cost for this deposit is determined by an interest rate $R$.\footnote{This interest rate corresponds to the difference of a funding rate minus the interest rate paid by the CCP for the deposit, typically $R\approx3\%$} Therefore, the IM cost can be modelled as an additional term in the dynamics of the self-financing portfolio as 
\begin{align}
\dV_{t} & =  \left(r\left(V_{t}-\pi_{t}\right)-R \CVaRaFt\left(V_{t}-V_{t+\Delta}\right)\right)\dt+\pi_t\frac{\dS_{t}}{S_{t}} \label{eq:1},
\end{align}
where the CVaR of a random variable $L$, conditional on the underlying sigma-field $\cF_t$ at time $t$, is defined by (see \cite{rock:urya:00})
\begin{equation}
\label{eq:cvar}
\CVaRaFt(L)=\inf_{x\in \R}\Esp{\frac{(L-x)^+}{1-\alpha}+x \Big| \mathcal{F}_{t}}.
\end{equation}
Since $V_{t+\Delta}$ may be meaningless as $t$ gets close to $T$, in \eqref{eq:1} one should consider $V_{(t+\Delta)\land T}$ instead. Rewriting \eqref{eq:1} in integral form together with the replication constraint, we obtain a BSDE
\begin{align}
V_t & =  \Phi(S_T)+\int_t^T\left(-r\left(V_{s}-\pi_{s}\right)-\mu\pi_{s}+R \CVaRaFs\left(V_{s}-V_{(s+\Delta)\land T}\right)\right)\ds-\int_t^T \pi_s\sigma\dW_s,\ t\in[0,T]\label{eq:portfolio}.
\end{align}
The conditional CVaR term is anticipative and non-linear in the sense
of McKean \cite{mcke:66}, for it involves the law of future variations of the portfolio
conditional to the knowledge of the past. This is an example of McKean Anticipative BSDE, which we study in broader generality in Section \ref{section:A general MKABSDE}.

Coming back to the financial setting, $(V,\pi)$ stands for a valuation rule which treats the IM adjustment as endogenous (in the sense that CVaR is computed on $V$ itself).
One could alternatively consider that CVaR is related to an exogenous valuation (the so-called \emph{mark-to-market}), for instance the one due to \eqref{eq:0} (assuming that \eqref{eq:0} models the market evolution of the option price). 
Later in Section \ref{section:The Case of CVaR variation margin}, we give quantitative error bounds between these different valuation rules.
Without advocating one with respect to the other, we rather compare their values and estimate (theoretically and numerically) how well one of their output prices approximates the others. As a consequence, these results may serve as a support for banks and regulators for improving risk management and margin requirement rules.

\subsection{Literature review on anticipative BSDEs and comparison with our contribution}
BSDEs were introduced by Pardoux and Peng \cite{pp1990}.
Since then, the theoretical properties of BSDEs with different generators and terminal conditions have been extensively studied.
The link between Markovian BSDEs and partial differential equations (PDEs) was studied in \cite{pp1992}. Under some smoothness assumptions, \cite{pp1992} established that the solution of the Markovian BSDE corresponds to the solution of a semi-linear parabolic PDE.
In addition, several applications in finance have been proposed, in particular by El Karoui and co-authors \cite{kpq1997} who considered the application to European option pricing in the constrained case. 
In fact, \cite{kpq1997} showed that, under some constraints on the hedging strategy, the price of a contingent claim is given by the solution of a non-linear convex BSDE.

Recently, a new class of BSDEs called anticipated\footnote{We equivalently use the word anticipated or anticipative in this work.} BSDEs (ABSDEs for short) was introduced by Peng and Yang \cite{py2009}.  The main feature of this class is that the generator includes not only the value of the solution at the present, but also at a future date. In \cite{py2009} the existence, uniqueness and a comparison theorem for the solution is provided under a kind of Lipschitz condition which depends on the conditional expectation. One can also find more general formulations of ABSDE in Cheredito and Nam \cite{cn2017}.
As in the case of classical BSDEs, the question of weakening the Lipschitz condition considered in \cite{py2009} has  been tackled by Yang and Elliott \cite{ye2013}, who extended the existence theorem for ABSDEs from Lipschitz to continuous coefficients, and proved that the comparison theorem for anticipated BSDEs still holds. They also established a minimal solution.
 
At the same time, Buckdahn and Imkeller \cite{bi2009} introduced the so-called time-delayed BSDEs (see also Delong and Imkeller \cite{di2009, di2010}). 
As opposed to the ABSDEs of \cite{py2009}, in this case the generator depends on the values of the solution at the present and at past dates, weighted with a time delay function. 
Assuming that the generator satisfies a certain kind of Lipschitz assumption depending on a probability measure, Delong and Imkeller \cite{di2010} proved the existence and uniqueness of a solution for a sufficiently small time horizon or for a sufficiently small Lipschitz constant of the generator.
These authors also showed that, when the generator is independent of $y$ and for a small delay, existence and uniqueness hold for an arbitrary Lipschitz constant. Later, Delong and Imkeller \cite{delong2012} provided an application of time-delayed BSDEs to problems of pricing and hedging, and portfolio management.
This work focuses on participating contracts and variable annuities, which are worldwide life insurance products with capital protections, and on claims based on the performance of an underlying investment portfolio. 

\red{More recently, Cr\'{e}pey \emph{et al.} \cite{cess2017} have worked in a setting which is close to the problem we tackle here, introducing an application of ABSDEs to the problem of computing different types of valuation adjustments (XVAs) for derivative prices, related to funding (X=F), capital (X=K) and credit risk (X=C). 
In particular, they focus on the case where the variation margins of an OTC contract can be funded directly with the capital of the bank involved in the trade, giving rise to different terms in the portfolio evolution equation.} 
The connection of economic capital and funding valuation adjustment leads to an ABSDE, whose anticipated part consists of a conditional risk measure over the martingale part of the portfolio on a future time period. 
These authors have showed that the system of ABSDEs for the FVA and the KVA processes is well-posed.
Mathematically, the existence and uniqueness of the solution to the system is  established together with the convergence of Picard iterations.

Motivated by the dynamics of the self-financing portfolio \eqref{eq:portfolio}, we consider a type of ABSDEs (a McKean ABSDEs) where the generator depends on the value of the solution, but also on the law of the future trajectory, possibly up to maturity, in analogy with \cite{cess2017}.
We state a priori estimates on the differences between the solutions of two such MKABSDEs. Based on these estimates, we derive existence and uniqueness results via a fixed-point theorem.  \red{Comparing again with the close work \cite{cess2017}, one main difference is that the time horizon $\Delta$ in the IM setting is around one week, as apposed to one year for economic capital in \cite{cess2017}.
Consequently, the IM framework naturally lends itself to an asymptotic analysis as $\Delta$ goes to zero, giving rise to the approximating non linear or linear (non anticipative) BSDEs for which we provide error estimates and explore different numerical methods.}

\section{A general McKean Anticipative BSDE}
\label{section:A general MKABSDE}
In order to give meaning to \eqref{eq:portfolio} and to more general (multidimensional) cases such as \eqref{eq:CvarBSDE} below, we now introduce  a general mathematical setup for studying existence and uniqueness of solutions.

\subsection{Notation}
Let $T>0$ be the finite time horizon and let $(\Omega, \cF,\P)$ be a probability space equipped with a $d$-dimensional Brownian motion, where $d\geq1$. 
We denote $(\cF_t)_{t\in[0,T]}$ the filtration generated by $W$, completed with the $\P$-null sets of $\cF$.
Let $t\in[0,T],\beta \geq 0$ and $m\in\mathbb{N}^{*}$.
We will make use of the following notations:
\begin{itemize}
\item For any $a=\left(a_{1},...,a_{m}\right)\in\R^{m}$, $|a|=\sqrt{\sum_{i=1}^{m}a_{i}^{2}}$.
\item Given a process $\left(x_{s}\right)_{s\in[0,T]}$, we set $x_{t:T}:=\left(x_{s}\right)_{s\in[t,T]}$.
\item $L_{T}^{2}\left(\R^{m}\right)=\left\{ \R^{m}\text{-valued \ensuremath{\cF_{T}}-measurable random variables \ensuremath{\xi} such that \ensuremath{\Esp{|\xi|^{2}}<\infty}}\right\}.$
\item $\bH_{\beta,T}^{2}\left(\R^{m}\right)=\left\{ \R^{m}\text{-valued and \ensuremath{\cF}-adapted stochastic processes \ensuremath{\varphi}\ such that \ensuremath{\Esp{\int_{0}^{T} e^{\beta t}|\varphi_{t}|^{2}\dt}<\infty}}\right\}. $ For $\varphi\in\bH_{\beta,T}^{2}\left(\R^{m}\right)$,
we define $||\varphi||_{\bH^2_{\beta, T}}=\sqrt{\Esp{\int_{0}^{T}e^{\beta t}|\varphi_{t}|^{2}\dt}}$.
\item $\bS_{\beta,T}^{2}\left(\R^{m}\right)=\left\{ \text{Continuous processes \ensuremath{\varphi\in\text{\ensuremath{\bH}}_{\beta ,T}^{2}\left(\R^{m}\right)}\ such that \ensuremath{\Esp{\sup_{t\in[0,T]}e^{\beta t}|\varphi_{t}|^{2}}<\infty}}\right\} .$ For $\varphi\in\bS_{\beta,T}^{2}\left(\R^{m}\right)$,
we define $||\varphi||_{\bS^2_{\beta ,T}}=\sqrt{\Esp{\sup_{t\in[0,T]} e^{\beta t}|\varphi_{t}|^{2}}}$.\end{itemize} 
Note that $\bH_{\beta, T}^{2}\left(\R^{m}\right)=\bH_{0,T}^{2}\left(\R^{m}\right)$ and $\bS_{\beta, T}^{2}\left(\R^{m}\right)=\bS_{0,T}^{2}\left(\R^{m}\right)$, for any $\beta \ge 0$.
The additional degree of freedom given by the parameter $\beta$ in the definition of the space norm will be useful when deriving a priori estimates (see Lemma \ref{lem:NEKestimates}).

\subsection{Main result}
Our aim is to find a pair of processes $\left(Y,Z\right)\in\bS_{0, T}^{2}\left(\R\right)\times\bH_{0, T}^{2}\left(\R^{d}\right)$
satisfying 
\begin{eqnarray}
Y_{t}  =  \xi+\int_{t}^{T}f\left(s,Y_{s},Z_{s},\Lambda_{s}\left(Y_{s:T}\right)\right)\ds-\int_{t}^{T}Z_{s}\dW_{s},
\qquad t\in[0,T],
\label{eq:BSDE}
\end{eqnarray}
for a certain mapping $\Lambda_t(\cdot)$ to be defined below.
We call Equation \eqref{eq:BSDE} McKean Anticipative BSDE (MKABSDE)
with parameters $\left(f,\Lambda, \xi \right)$. 
In order to obtain existence and uniqueness of solutions, we require that the mappings $f$ and $\Lambda$ satisfy some suitable Lipschitz properties (specified below), and that the terminal condition $\xi$  be square integrable.

\begin{assumption*}[S] For any $y,z,\lambda\in\R\times\R^{d}\times\R$,
$f(\cdot,y,z,\lambda)$ is a $\cF$-adapted stochastic process with values in $\R$ and there exists a constant $C_{f}>0$ such that almost surely, for all
$(s,y_{1},z_{1},\lambda_{1}),(s,y_{2},z_{2},\lambda_{2})\in[0,T]\times\R\times\R^{d}\times\R$,
\[
\left|f\left(s,y_{1},z_{1},\lambda_{1}\right)-f\left(s,y_{2},z_{2},\lambda_{2}\right)\right|\leq C_{f}\left(|y_{1}-y_{2}|+|z_{1}-z_{2}|+|\lambda_{1}-\lambda_{2}|\right).
\]
Moreover, $\Esp{\int_{0}^{T}\left|f\left(s,0,0,0\right) \right|^{2} \ds}<\infty.$
\end{assumption*}

\begin{assumption*}[A] For any $X\in\bS_{0, T}^{2}\left(\R\right)$, $
\left(\Lambda_t\left(X_{t:T}\right)\right)_{t\in[0,T]}$ defines a stochastic process that belongs to 
$\bH_{0,T}^{2}(\R)$. 
There exist a constant $C_{\Lambda}>0$ and a family of \red{positive} measures $\left(\nu_{t}\right)_{t\in[0,T]}$ on $\R$ such that for every $t\in[0,T]$, $\nu_{t}$ has support included in $[t,T]$, $\nu_t \left([t,T]\right)=1$, and for any $y^{1},y^{2}\in\bS_{0,T}^{2}\left(\R\right)$,
we have 
\begin{align*}
\left|\Lambda_{t}\left(y_{t:T}^{1}\right)-\Lambda_{t}\left(y_{t:T}^{2}\right)\right|
& \leq
C_{\Lambda} \, \Esp{\int_{t}^{T}\left|y_{s}^{1}-y_{s}^{2}\right|\nu_{t}\left(\ds\right)\Big| \cF_{t}},\dt\otimes \dd \P\ a.e.\ .
\end{align*}
\end{assumption*}

\red{\noindent Note that, under Assumption (A), for every $\beta\geq0$ and every continuous path $x:[0,T]\rightarrow\R$ we have
\[
\int_{0}^{T}e^{\beta s}\int_{s}^{T}\left|x_{u}\right|\nu_{s}\left(\du\right)\ds
\leq T \sup_{t\in[0,T]} e^{\beta t}|x_t|.
\]}
We will say that a function $\tilde{f}$ (resp.\til a mapping $\tilde{\Lambda}$) satisfies Assumption (S) (resp. (A)) if the assumption holds for the choice $f = \tilde{f}$ (resp.\til$\Lambda = \tilde{\Lambda}$). We can now give the main result of this section.
\begin{thm}
\label{thm:MKABSDE}Under Assumptions \emph{(S)} and \emph{(A)}, for any terminal condition $\xi\in L_{{T}}^{2}\left(\R\right)$ the BSDE \eqref{eq:BSDE} has a unique solution $\left(Y,Z\right)\in\bS_{0, T}^{2}\left(\R\right)\times\bH_{0, T}^{2}\left(\R^{d}\right)$.
\end{thm}

\subsection{Proof of Theorem \ref{thm:MKABSDE}}
The proof uses classical arguments. We first establish apriori estimates
in the same spirit as in \cite{kpq1997} on the solutions to the BSDE.
Then for a suitable constant $\beta\geq0$, we use Picard's fixed point method in the space $\bS_{\beta, T}^{2}\left(\R\right)\times\bH_{\beta, T}^{2}\left(\R^{d}\right)$
to obtain existence and uniqueness of a solution to Equation \eqref{eq:BSDE}.

\begin{lem}\label{lem:NEKestimates}
Let $\left(Y^{1},Z^{1}\right),\left(Y^{2},Z^{2}\right)\in\bS_{0,T}^{2}\left(\R\right)\times\bH_{0,T}^{2}\left(\R^{d}\right)$
be solutions to MKABSDE \eqref{eq:BSDE} associated respectively
to the parameters $\left(f^{1},\Lambda^{1},\text{\ensuremath{\xi}}^{1}\right)$
and $\left(f^{2},\Lambda^{2},\text{\ensuremath{\xi}}^{2}\right)$. We assume that $f^1$ satisfies Assumption $(S)$ and that $\Lambda^1$ satisfies Assumption $(A)$. Let us define $\delta Y:=Y^{1}-Y^{2}$, $\delta Z:=Z^{1}-Z^{2}$, $\delta \xi:=\xi^{1}-\xi^{2}$. Finally, let us define for
$s\in[0,T]$, $$\delta_{2}f_{s}=f^{1}\left(s,Y_{s}^{2},Z_{s}^{2},\Lambda^{2}\left(Y_{s:T}^{2}\right)\right)-f^{2}\left(s,Y_{s}^{2},Z_{s}^{2},\Lambda^{2}\left(Y_{s:T}^{2}\right)\right),\quad\text
{and}\quad \delta_{2}\Lambda_{s}=\Lambda_{s}^{1}\left(Y_{s:T}^{2}\right)-\Lambda_{s}^{2}\left(Y_{s:T}^{2}\right).$$
Then there exists a constant $C>0$ such that for $\mu > 0$, we have for $\beta$ large enough
\begin{align*}
||\delta Y||_{\bS_{\beta,T}^{2}}^2 & \leq  C\left(e^{\beta T}\Esp{|\delta \xi|^{2}}+\frac{1}{\mu^{2}}\left(||\delta_{2}f||_{\bH_{\beta,T}^{2}}^{2}+C_{f^1}||\delta_{2}\Lambda||_{\bH_{\beta,T}^{2}}^{2}\right)\right),\\
||\delta Z||_{\bH_{\beta,T}^{2}}^2 & \leq  C\left(e^{\beta T}\Esp{|\delta \xi|^{2}}+\frac{1}{\mu^{2}}\left(||\delta_{2}f||_{\bH_{\beta,T}^{2}}^{2}+C_{f^1}||\delta_{2}\Lambda||_{\bH_{\beta,T}^{2}}^{2}\right)\right).
\end{align*}
 \end{lem}
\begin{proof}
The proof is based on similar arguments used in \cite{kpq1997}. Let us use the decomposition:
\eqstar{
|f^1(s,Y^1_s,Z^1_s, \Lambda^1_s(Y^1_{s:T}))& - f^2(s,Y^2_s,Z^2_s, \Lambda^2_s(Y^2_{s:T}))|\\
 &\leq  \left|f^1(s,Y^1_s,Z^1_s, \Lambda^1_s(Y^1_{s:T})) - f^1(s,Y^2_s,Z^2_s, \Lambda^2_s(Y^2_{s:T})) \right| \\
&\quad +\left|f^1(s,Y^2_s,Z^2_s, \Lambda^2_s(Y^2_{s:T})) - f^2(s,Y^2_s,Z^2_s, \Lambda^2_s(Y^2_{s:T}))\right|\\
& \leq   C_{f^1}\left(|\delta Y_s|+ |\delta Z_s|+ | \Lambda^1_s(Y^1_{s:T})-  \Lambda^2_s(Y^2_{s:T})|\right)+ |\delta_2 f_s |\\
& \leq   C_{f^1}\left(|\delta Y_s|+ |\delta Z_s|+| \Lambda^1_s(Y^1_{s:T})-  \Lambda^1_s(Y^2_{s:T})| +|\delta_2 \Lambda_s| \right)+ |\delta_2 f_s |.
}
By It\^{o}'s lemma on the process $t \rightarrow e^{\beta t} |\delta Y_t|^2$, where $\beta \geq 0$, and using the previous inequality, we have that
\begin{align}
& e^{\beta t} |\delta Y_t|^2 + \beta \int_t^T e^{\beta s} |\delta Y_s|^2 \ds + \int_t^T e^{\beta s} |\delta Z_s |^2 \ds\nonumber   \\
& = e^{\beta T}|\delta \xi|^2 + 2\int_t^T e^{\beta s} \delta Y_s \left( f^1(s,Y^1_s,Z^1_s, \Lambda^1_s(Y^1_{s:T})) - f^2(s,Y^2_s,Z^2_s, \Lambda^2_s(Y^2_{s:T})) \right) \ds \nonumber - 2\int_t^T e^{\beta s} \delta Y_s \delta Z_s \dW_s \nonumber \\
& \leq e^{\beta T}|\delta \xi|^2 + 2 \int_t^T e^{\beta s} |\delta Y_s| \left( C_{f^1}\left( |\delta Y_s|+ |\delta Z_s|+  | \Lambda^1_s(Y^1_{s:T})-  \Lambda^1_s(Y^2_{s:T})| +|\delta_2 \Lambda_s|\right) + |\delta_2 f_s | \right) \ds \nonumber \\
&  \quad- 2\int_t^T e^{\beta s} \delta Y_s \delta Z_s \dW_s.\label{eq:estimates2}
\end{align}
Applying Young's inequality with $\lambda, \mu \neq 0$, we have
\begin{align*}
 2 |\delta Y_s|&\left( C_{f^1}(|\delta Z_s| + | \Lambda^1_s(Y^1_{s:T})-  \Lambda^1_s(Y^2_{s:T})| + |\delta_2 \Lambda_s|) + |\delta_2 f_s |\right)\\
& \leq \frac{C_{f^1}}{\lambda^2}|\delta Z_s|^2 + \lambda^2 C_{f^1} |\delta Y_s|^2 + \frac{C_{f^1}}{\lambda^2}|\Lambda^1_s(Y^1_{s:T})-  \Lambda^1_s(Y^2_{s:T})|^2 + \lambda^2 C_{f^1} |\delta Y_s|^2\\
&  \quad+ \frac{C_{f^1}}{\mu^2} |\delta_2 \Lambda_s|^2 + \mu^2 C_{f^1} |\delta Y_s|^2 + \frac{1}{\mu^2} |\delta_2 f_s|^2 + \mu^2  |\delta Y_s|^2\\
& \leq \Big( \mu^2 + C_{f^1}(\mu^2 +  2\lambda^2) \Big) |\delta Y_s|^2 +\frac{C_{f^1}}{\lambda^2}|\delta Z_s|^2 +\frac{C_{f^1}}{\lambda^2}|\Lambda^1_s(Y^1_{s:T})-  \Lambda^1_s(Y^2_{s:T})|^2 \\
& \quad+ \frac{C_{f^1}}{\mu^2} |\delta_2 \Lambda_s|^2  +\frac{1}{\mu^2} |\delta_2 f_s|^2.
\end{align*}
Then plug this bound into \eqref{eq:estimates2} to get
\begin{align}
 e^{\beta t} |\delta Y_t|^2 &+ \beta \int_t^T e^{\beta s} |\delta Y_s|^2 \ds + \int_t^T e^{\beta s} |\delta Z_s |^2 \ds  \nonumber \\
& \leq e^{\beta T}|\delta \xi|^2 +  \Big( \mu^2 + C_{f^1}(2+ \mu^2 + 2 \lambda^2) \Big) \int_t^T e^{\beta s} |\delta Y_s|^2 \ds + \frac{C_{f^1}}{\lambda^2} \int_t^T e^{\beta s} |\delta Z_s|^2 \ds \nonumber\\
&  \quad+ \frac{C_{f^1}}{\lambda^2} \int_t^T e^{\beta s} |\Lambda^1_s(Y^1_{s:T})-  \Lambda^1_s(Y^2_{s:T})|^2 \ds + \frac{C_{f^1}}{\mu^2} \int_t^T e^{\beta s} |\delta_2 \Lambda_s|^2 \ds  +\frac{1}{\mu^2} \int_t^T e^{\beta s} |\delta_2 f_s|^2 \ds \nonumber\\
&\quad- 2\int_t^T e^{\beta s} \delta Y_s \delta Z_s \dW_s.\label{eq:estimates3}
\end{align}
Choosing $\lambda^2 > C_{f^1}$ and 
\begin{equation}
\label{eq:beta:choice}
\beta \geq \mu^2 + C_{f^1}(2+ \mu^2 + 2 \lambda^2),
\end{equation}
we get from \eqref{eq:estimates3} that
\begin{align}
\Esp{\int_t^T e^{\beta s} |\delta Z_s |^2 \ds}
 &\leq \frac{\lambda^2}{\lambda^2 - C_{f^1}}\Esp{ e^{\beta T}|\delta \xi|^2   + \frac{C_{f^1}}{\lambda^2} \int_t^T e^{\beta s} |\Lambda^1_s(Y^1_{s:T})-  \Lambda^1_s(Y^2_{s:T})|^2 \ds} \nonumber \\
& \quad+\frac{\lambda^2}{\lambda^2 - C_{f^1}} \Esp{ \frac{C_{f^1}}{\mu^2} \int_t^T e^{\beta s} |\delta_2 \Lambda_s|^2 \ds  +\frac{1}{\mu^2} \int_t^T e^{\beta s} |\delta_2 f_s|^2 \ds}.\nonumber
\end{align}
Here we have used that the stochastic integral in \eqref{eq:estimates3} is a true martingale, by invoking $\delta Y\in \bS_{0,T}^{2}\left(\R\right)$, $\delta Z \in \bH_{0,T}^{2}\left(\R^{d}\right)$, computations similar to \eqref{eq:estimates5bis} and a localization procedure.
From  \eqref{eq:estimates3} we also have that
\begin{align}
&\Esp{ \sup_{t\in[0,T]} e^{\beta t} |\delta Y_t|^2 +  \left(1- \frac{C_{f^1}}{\lambda^2}\right) \int_0^T e^{\beta s} |\delta Z_s |^2 \ds } \nonumber\\
& \leq \E \Bigg[e^{\beta T}|\delta \xi|^2 + \frac{C_{f^1}}{\lambda^2} \int_0^T e^{\beta s} |\Lambda^1_s(Y^1_{s:T})-  \Lambda^1_s(Y^2_{s:T})|^2 \ds+ \frac{C_{f^1}}{\mu^2} \int_0^T e^{\beta s} |\delta_2 \Lambda_s|^2 \ds + \frac{1}{\mu^2} \int_0^T e^{\beta s} |\delta_2 f_s|^2 \ds\nonumber \\
&  \qquad + 2 \sup_{t\in[0,T]} \left|\int_t^T e^{\beta s}  \delta Y_s \delta Z_s \dW_s \right|\Bigg].\label{eq:estimates1}
\end{align}
As $\Lambda^1$ satisfies Assumption (A), Jensen's inequality yields
\begin{align}
 \Esp{ \int_0^T e^{\beta s} |\Lambda^1_s(Y^1_{s:T})-  \Lambda^1_s(Y^2_{s:T})|^2 \ds}   &\leq C^2_{\Lambda} \Esp{ \int_0^T e^{\beta s} \int_s^T |\delta Y_u|^2 \nu_s(\du) \ds} \nonumber\\
 &\leq T C^2_{\Lambda} \Esp{ \sup_{t\in[0,T]} e^{\beta t}|\delta Y_t|^2 }. \label{eq:estimates5}
\end{align}
By the Burkholder-Davis-Gundy inequality, there exists a positive constant $C_1$ such that
\begin{align}
\Esp{\sup_{t\in[0,T]} \Big|\int_t^T e^{\beta s}  \delta Y_s \delta Z_s \dW_s \Big|} & \leq C_1 \Esp{\left(\int_0^T e^{2\beta s}  |\delta Y_s|^2 |\delta Z_s|^2 \ds \right)^{1/2}}\nonumber\\
& \leq C_1 \Esp{\left(\sup_{s\in[0,T]} e^{\beta s}  |\delta Y_s|^2 \right)^{1/2} \left(\int_0^T e^{\beta s}  |\delta Z_s|^2 \ds \right)^{1/2}}.
\label{eq:estimates5bis}
\end{align}
Therefore, by Young's inequality with $\gamma > 0$, we have
\begin{align}
2 \Esp{ \sup_{t\in[0,T]} \left|\int_t^T e^{\beta s}  \delta Y_s \delta Z_s \dW_s \right|}& \leq \frac{C_1}{\gamma^2}\Esp{ \sup_{t\in[0,T]} e^{\beta t}  |\delta Y_t|^2} + \gamma^2 C_1 \Esp{ \int_0^T e^{\beta s}  |\delta Z_s|^2 \ds }\nonumber \\
& \leq \frac{C_1}{\gamma^2}\Esp{ \sup_{t\in[0,T]} e^{\beta t}  |\delta Y_t|^2 }
+ \frac{\gamma^2 C_1   \lambda^2}{\lambda^2 - C_{f^1}}\E\Bigg[ e^{\beta T}|\delta \xi|^2 + \frac{C_{f^1}}{\mu^2} \int_t^T e^{\beta s} |\delta_2 \Lambda_s|^2 \ds   \nonumber \\
&\quad
+ \frac{1}{\mu^2} \int_t^T e^{\beta s} |\delta_2 f_s|^2 \ds \Biggr]
+ \frac{\gamma^2 C_1   \lambda^2}{\lambda^2 - C_{f^1}}\frac{C_{f^1}}{\lambda^2} 
T C^2_{\Lambda} \Esp{ \sup_{t\in[0,T]} e^{\beta t}|\delta Y_t|^2  }.
\label{eq:estimates4}
\end{align}
Combining the inequalities \eqref{eq:estimates1}--\eqref{eq:estimates4} leads to
\begin{align*}
& \left(1 - \frac{C_1}{\gamma^2} -  \frac{T C_{f^1} C^2_{\Lambda}}{\lambda^2} - \frac{T C_{f^1}  C_1 C^2_{\Lambda} \gamma^2  }{\lambda^2 - C_{f^1}} \right)\Esp{ \sup_{t\in[0,T]} e^{\beta t} |\delta Y_t|^2}+  \Big(1- \frac{C_{f^1}}{\lambda^2}\Big) \Esp{\int_0^T e^{\beta s} |\delta Z_s |^2 \ds}\\
& \leq \Big( 1+  \frac{\gamma^2 C_1 \lambda^2}{\lambda^2 - C_{f^1}} \Big) \left( \Esp{ e^{\beta T}|\delta \xi|^2} + \frac{1}{\mu^2} \Esp{ C_{f^1}\int_0^T e^{\beta s} |\delta_2 \Lambda_s|^2 \ds  + \int_0^T e^{\beta s} |\delta_2 f_s|^2 \ds}\right).
\end{align*}
Let us define a continuous function $\Gamma$ by setting
\[
\Gamma(\gamma, \lambda)=1 - \frac{C_1}{\gamma^2} -  \frac{T C_{f^1} C^2_{\Lambda}}{\lambda^2} - \frac{T C_{f^1}  C_1 C^2_{\Lambda} \gamma^2  }{\lambda^2 - C_{f^1}}
\]
for any $\gamma>0$ and any $\lambda >0$ with $\lambda^2> C_{f^1}$. Observe that  if we set $\gamma(\lambda) = \sqrt{\lambda}$ with $\lambda>0$, we have $\lim_{\lambda\rightarrow\infty} \Gamma(\gamma(\lambda),\lambda) = 1$, so that there exist $\lambda,\gamma$ large enough  such that $\Gamma(\gamma,\lambda) > 0$. For such a choice of $\gamma$ and $\lambda$, we obtain the announced result with the constant
\[
C = \frac{ 1+  \frac{C_1 \gamma^2  \lambda^2}{\lambda^2 - C_{f^1}}}{\min\left(\Gamma(\gamma,\lambda), 1-\frac{C_{f^1}}{\lambda^2}\right)}.
\]
Recall that $\beta$ is chosen according to the value of $\lambda$ (see inequality \eqref{eq:beta:choice}).
\end{proof}

\noindent  \textbf{Proof of Theorem \ref{thm:MKABSDE}.}
We use the previous apriori estimates in the case where $(Y^1,Z^1)$ and $(Y^2,Z^2)$ solve respectively the BSDEs
\begin{align*}
Y_{t}^{1} & =  \xi+\int_{t}^{T}f_{s}\left(U_{s}^{1},V_{s}^{1},\Lambda_{s}\left(U_{s:T}^{1}\right)\right)\ds-\int_{t}^{T}Z_{s}^{1}\dW_{s},\\
Y_{t}^{2} & =  \xi+\int_{t}^{T}f_{s}\left(U_{s}^{2},V_{s}^{2},\Lambda_{s}\left(U_{s:T}^{2}\right)\right)\ds-\int_{t}^{T}Z_{s}^{2}\dW_{s}.
\end{align*}
Here, $(U^1,V^1),(U^2,V^2)\in\bS_{0,T}^{2}\left(\R\right)\times\bH_{0,T}^{2}\left(\R^{d}\right)$ are given processes. Therefore $f_{s}\left(U_{s}^{1},V_{s}^{1},\Lambda_{s}\left(U_{s:T}^{1}\right)\right)$ and $f_{s}\left(U_{s}^{2},V_{s}^{2},\Lambda_{s}\left(U_{s:T}^{2}\right)\right)$ define processes in $\bH_{0,T}^{2}\left(\R\right)$ owing to the Assumptions (S) and (A). Therefore, the existence and uniqueness of   $(Y^1,Z^1)$ and $(Y^2,Z^2)$ in $\bS_{0, T}^{2}\left(\R\right)\times\bH_{0, T}^{2}\left(\R^{d}\right)$ as solutions of BSDEs is classical  (see \cite[Theorem 2.1, Proposition 2.2]{kpq1997}).

In addition, the process $Y^{1}-Y^{2}$ is then solution to the BSDE 
$
Y_{t}^{1}-Y_{t}^{2}=\int_{t}^{T}\delta_2 f_{s}\ds-\int_{t}^{T}\left(Z_{s}^{1}-Z_{s}^{2}\right)\dW_{s},
$
where the driver $\delta_2 f_{s}=f_{s}\left(U_{s}^{1},V_{s}^{1},\Lambda_{s}\left(U_{s:T}^{1}\right)\right)-f_{s}\left(U_{s}^{2},V_{s}^{2},\Lambda_{s}\left(U_{s:T}^{2}\right)\right)$
does not depend on $Y_{s}^{1}$ nor $Y_{s}^{2}$. Using Lemma \ref{lem:NEKestimates} for $C_f = 0$ and $\mu > 0$, we have that for $\beta > 0$ large enough, 
$$||\delta Y||^2_{\bS^2_{\beta,T}}+||\delta Z||^2_{\bH^2_{\beta,T}}\leq \frac{C}{\mu^2}||\delta_2 f||^2_{\bH^2_{\beta,T}}.$$
Moreover,  
\begin{align*}
||\delta_2 f||^2_{\bH^2_{\beta,T}} & =  \Esp{\int_{0}^{T}e^{\beta s}|f_{s}\left(U_{s}^{1},V_{s}^{1},\Lambda_{s}\left(U_{s:T}^{1}\right)\right)-f_{s}\left(U_{s}^{2},V_{s}^{2},\Lambda_{s}\left(U_{s:T}^{2}\right)\right)|^{2}\ds}\\
 & \leq  3C_{f}^2\Esp{\int_{0}^{T}e^{\beta s}\left(|\delta U_{s}|^{2}+|\delta V_{s}|^{2}+|\Lambda_{s}\left(U_{s:T}^{1}\right)-\Lambda_{s}\left(U_{s:T}^{2}\right)|^{2}\right)\ds}.
\end{align*}
As we have that $||\delta U||^2_{\bH^2_{\beta,T}}\leq T||\delta U||^2_{\bS^2_{\beta,T}}$, and
\begin{align*}
\Esp{\int_{0}^{T}e^{\beta s}|\Lambda_{s}\left(U_{s:T}^{1}\right)-\Lambda_{s}\left(U_{s:T}^{2}\right)|^{2}\ds} & \leq  C_{\Lambda}^2 \Esp{\int_{0}^{T}e^{\beta s}\int_{s}^{T}|\delta U_{u}|^{2}\nu_{s}\left(\du\right)\ds}
\leq T C_{\Lambda}^2 ||\delta U||^2_{\bS^2_{\beta,T}}.
\end{align*}
We obtain that 
\[
||\delta Y||_{\bS_{\beta,T}^{2}}^{2}+||\delta Z||_{\bH_{\beta,T}^{2}}^{2}\leq \frac{3CC_{f}^2}{\mu^2}\left(\left(T C_{\Lambda}^2+T\right)||\delta U||_{\bS_{\beta,T}^{2}}+||\delta V||_{\bH_{\beta,T}^{2}}^{2}\right).
\]
We now choose $\mu^2 > 3CC_{f}^2(T C^2_{\Lambda} + T + 1)$, and obtain that for $\beta$ large enough, the mapping
$\phi:\left(U,V\right)\rightarrow\left(Y,Z\right)$ is a contraction in the space ${\bS_{\beta,T}^{2}}(\R)\times {\bH_{\beta,T}^{2}}(\R^d)$. Hence, we get existence and uniqueness of a solution to the BSDE \eqref{eq:BSDE}. \qed

\section{The Case of CVaR initial margins}
\label{section:The Case of CVaR variation margin}
In this section, we apply the previous results on MKABSDE to a generalization of equation \eqref{eq:portfolio}.
We will consider a multi-dimensional setting (several assets $S^i$), a more general dynamic model and a more general terminal condition: we refer to the following section for precise definitions.
Beyond usual existence and uniqueness results, our aim is to analyse related approximations of the option price and delta, obtained when the CVaR is evaluated using Gaussian expansions (justified as $\Delta\to0$, see Theorem \ref{theorem: approximation:price}).

\subsection{Well-posedness of the problem}
Let us consider a general It\^o market with  $d$ tradable assets \cite[Chapter 1]{kara:shre:98}. The riskless asset $S^{0}$ (money account) follows
the dynamics $\frac{\dS_{t}^{0}}{S_{t}^{0}}=r_{t}\dt,$ and we have
$d$ risky assets $\left(S^{1},...,S^{d}\right)$ following 
\begin{equation}
\label{eq:ito market}
\frac{\dS_{t}^{i}}{S_{t}^{i}}=\mu_{t}^{i}\dt+\sum_{j=1}^{d}\sigma_{t}^{ij}\dW_{t}^{j},\ S_{0}^{i}=s_{0}^{i}\in\R,\ 1\leq i\leq d.
\end{equation}
The processes $r,\mu:=\left(\mu^{i}\right)_{1\leq i\leq d},\sigma:=\left(\sigma^{ij}\right)_{1\leq i,j\leq d}$
are $\cF$-adapted stochastic processes with values respectively in $\R,\R^{d},$
and the set of matrices of size $d\times d$. Moreover, we assume
that $\dt\otimes \dd\P$ a.e., the matrix $\sigma_{t}$ is invertible
and the processes $r$ and $\sigma^{-1}\left(\mu-r\1\right)$
are uniformly bounded, where we define the column vector $\1:=\left(1,...,1\right)^\top\in\R^{d}$. For a path-dependent payoff $\xi$ paid
at maturity $T$, the dynamics of the hedging portfolio $\left(V,\pi\right)$
with CVaR initial margin requirement (over a period $\Delta>0$) is given by
\begin{equation*}
V_{t}=\xi+\int_{t}^{T}\left(-r_{s}V_{s}+\pi_{s}\left(r_{s}\1-\mu_{s}\right)+R\CVaRaFs\left(V_{s}-V_{\left(s+\Delta\right)\wedge T}\right)\right)\ds-\int_{t}^{T}\pi_{s}\sigma_{s}\dW_{s}.
\end{equation*}
Here $\pi$ is a row vector whose $i$th coordinate consists of the amount invested in the $i$th asset. The derivation is analogous to that of Section \ref{subsection:A first Anticipative BSDE with dependence in law}.
This equation rewrites, in terms of the variables $(V,Z=\pi\sigma)$, 
\begin{equation}
V_{t}=\xi+\int_{t}^{T}\left(-r_{s}V_{s}+Z_s\sigma_{s}^{-1}\left(r_{s}\1-\mu_{s}\right)+R\CVaRaFs\left(V_{s}-V_{\left(s+\Delta\right)\wedge T}\right)\right)\ds-\int_{t}^{T}Z_{s}\dW_{s}.\label{eq:CvarBSDE}
\end{equation}
Existence and uniqueness for the MKABSDE above are a consequence of Theorem \ref{thm:MKABSDE}.
\begin{corollary} For any square integrable terminal condition $\xi$, 
the CVaR initial margin problem \eqref{eq:CvarBSDE} is well posed with a unique solution $(V,Z)\in {\bS_{\beta,T}^{2}}(\R)\times {\bH_{\beta,T}^{2}}(\R^d)$ for any $\beta \ge 0$.
\end{corollary}

\red{\noindent Note that, when $R = 0$, the coefficients $r_t, \mu_t$ and $\sigma_t$ are deterministic functions of time and $\xi = \Phi(S_T)$, \eqref{eq:CvarBSDE} is solved by $V_t = v(t,S_t)$, where $S$ follows \eqref{eq:ito market} and $v$ is the solution to the classical (multi-dimensional) Black-Scholes pricing equation with interest rate $r_t$, volatility parameter $\sigma_t$ and payoff function $\Phi$.
As usual, the function $v$ does not depend on the drift parameter $\mu$, while the portfolio process $V_t = v(t,S_t)$ (aka the MtM of the trade) does, via the asset price $S$.
In general, the solution $(V, Z)$ to \eqref{eq:CvarBSDE} will still depend on $\mu$, even if we cannot identify the analogous of the function $v$ in the general case.}

\begin{proof}
We check that Assumptions S and A are satisfied for the problem \eqref{eq:CvarBSDE}.
The driver of the BSDE has the form
\[
f\left(t,v,z,\lambda\right)=-r_{t}v+z\sigma_{t}^{-1}\left(r_{t}\1-\mu_{t}\right)+\lambda,\ t\geq0,\ v,\lambda\in\R,\ z\in\R^{d},
\]
and we introduce the functional 
\[
\Lambda_{t}\left(X_{t:T}\right):=R\CVaRaFt\left(X_{t}-X_{\left(t+\Delta\right)\wedge T}\right)=R\inf_{x\in\R}\Esp{\frac{\left(X_{t}-X_{\left(t+\Delta\right)\wedge T}-x\right)^{+}}{1-\alpha}+x\mid \cF_{t}},\ t\in[0,T],\ X\in\bS_{0,T}^{2}(\R).
\]
Since $r$ and $\sigma^{-1}\left(\mu-r\1\right)$
are uniformly bounded, $f$ clearly satisfies Assumption (S).  We now check
that $\Lambda$ satisfies Assumption (A). 
For $X\in\bS_{0,T}^{2}(\R)$ and $x\in\mathbb{R},$ we have  
\begin{equation}
\label{eq:ineqLambda}
\mathbb{E}\left[\left.X_t - X_{(t+\Delta)\wedge T} \right| \mathcal{F}_t\right] \leq \inf_{x\in\mathbb{R}} \mathbb{E}\left[\left.\frac{(X_t - X_{(t+\Delta)\wedge T} - x)^+}{1 - \alpha} + x \right| \mathcal{F}_t\right] \leq \mathbb{E}\left[\left.\frac{(X_t - X_{(t+\Delta)\wedge T})^+}{1 - \alpha} \right| \mathcal{F}_t\right],
\end{equation}
where for the left hand side (l.h.s.)\til we use the fact that $\frac{(z-x)^+}{1-\alpha} + x \geq z$ for $\alpha\in(0,1)$ and $z,x\in\mathbb{R}$,  and for the right hand side (r.h.s.), we bound the infimum from above with the value taken at $x=0$.
Since it is easy to check that both the l.h.s.\til and the r.h.s.\til of \eqref{eq:ineqLambda} belong to $\mathbb{H}^2_{0,T}$, we conclude that $\Lambda_{t}\left(X\right)\in\bH_{0,T}^{2}(\R)$. Now, let $X^{1},X^{2}\in\bS_{0,T}^{2}$.
We have
\begin{align*}
|\Lambda_{t}\left(X^{1}\right)-\Lambda_{t}\left(X^{2}\right)|
& \leq 
R \, \Esp{\left|\frac{X_{t}^{1}-X_{t}^{2}-\left(X_{\left(t+\Delta\right)\wedge T}^{1}-X_{\left(t+\Delta\right)\wedge T}^{2}\right)}{1-\alpha}\right|\mid \cF_{t}}\\
&\leq\frac{2 R}{1-\alpha}\Esp{\int_{t}^{T}\left|X_{s}^{1}-X_{s}^{2}\right|\nu_{t}\left(\ds\right)\mid \cF_{t}},
\end{align*}
where, for the first inequality, we use the fact that $\left|\inf_{x\in\R}g^{1}(x)-\inf_{x\in\R}g^{2}(x)\right|\leq\sup_{x\in\R}\left|g^{1}(x)-g^{2}(x)\right|$
for any couple of functions $g^{1},g^{2}:\R\rightarrow\R$ and
the 1-Lipschitz property of the positive part function, and for the second inequality, \red{we have set $\nu_{t} := \frac 12 (\delta_{t} + \delta_{\left(t+\Delta\right)\wedge T})$, where $\delta_{u}$ is the Dirac measure on $\{u\}$.
We conclude that Assumption (A) holds with $C_{\Lambda}=\frac{2 R}{1-\alpha}$.} 
We finally apply Theorem \ref{thm:MKABSDE} to obtain the claim.
\end{proof}

\subsection{Approximation by standard BSDEs when $\Delta\ll 1$} \label{sec:approximationsBSDE}
The numerical solution of \eqref{eq:CvarBSDE} is challenging in full generality.
In fact, it is a priori more difficult than solving a standard BSDE, for which we can employ, for example, regression Monte-Carlo methods (see e.g.~\cite{gobe:turk:14} and references therein).
In this work, we take advantage of the fact that $\Delta$ is small (recall $\Delta=$ one week or 10 days) in order to provide handier approximations of $(V,Z)$ defined in terms of standard non-linear or linear BSDEs.
Below, we define these different BSDEs and  provide the related error estimates.

At the lowest order in the parameter $\sqrt{\Delta}$, formally we have that, conditionally to $\cF_{s}$,
\begin{equation*}
V_{s}-V_{\left(s+\Delta\right)\wedge T}\approx-\int_{s}^{\left(s+\Delta\right)\wedge T}Z_u\dW_{u}\overset{(d)}{=}-|Z_{s}|\sqrt{\left(s+\Delta\right)\wedge T-s}\times G,
\end{equation*}
where we freeze the process $Z$ at current time $s\in[0,T]$, and $G\overset{(d)}{=}\mathcal{N}\left(0,1\right)$ is independent from $\cF_{s}$.
This is an approximation of CVaR using the ``Delta'' portion of the portfolio (in the spirit of \cite[Section 2]{glas:heid:shah:00}). Plugging this approximation into (\ref{eq:CvarBSDE}),
and defining 
\begin{equation}
\label{eq:CVaR:gaussien}
C_{\alpha}:=\CVaRa\left(\mathcal{N}\left(0,1\right)\right)=\left.\frac{e^{-x^2/2}  }{ (1-\alpha)\sqrt {2\pi}}\right|_{x=\cN^{-1}(\alpha)},
\end{equation}
we obtain a standard non-linear BSDE 
\begin{align}
V_{t}^{NL}  =  \xi+\int_{t}^{T}\left(-r_{s}V_{s}^{NL}+Z_{s}^{NL}\sigma_{s}^{-1}\left(r_{s}\1-\mu_{s}\right)
+
R \, C_{\alpha}\sqrt{\left(s+\Delta\right)\wedge T - s}|Z_{s}^{NL}|\right)\ds-\int_{t}^{T}Z_{s}^{NL}\dW_{s}.\quad \label{eq:NLBSDE}
\end{align}
Seeing $V^{NL}$ as a function of the parameter $\Delta$ appearing in the driver, we follow the expansion procedure in \cite{GobetPagliarani} and perform an expansion at the orders $0$ and $1$ w.r.t.~$\sqrt \Delta$.
We obtain two linear BSDEs, respectively $(V^{BS},Z^{BS})$ and $(V^{L},Z^L)$, where
\begin{align}
V_{t}^{BS} & =  \xi+\int_{t}^{T}\left(-r_{s}V_{s}^{BS}+Z_{s}^{BS}\sigma_{s}^{-1}\left(r_{s}\1-\mu_{s}\right)\right)\ds-\int_{t}^{T}Z_{s}^{BS}\dW_{s},\label{eq:BSBSDE}\\
V_{t}^{L} & =  \xi+\int_{t}^{T}\left(-r_{s}V_{s}^{L}+Z_{s}^{L}\sigma_{s}^{-1}\left(r_{s}\1-\mu_{s}\right)+ R C_{\alpha}\sqrt{\left(s+\Delta\right)\wedge T - s}|Z_{s}^{BS}|\right)\ds-\int_{t}^{T}Z_{s}^{L}\dW_{s}.\label{eq:LBSDE}
\end{align}
Let us comment on these different models:
\begin{itemize}
\item The simplest equation is $(V^{BS},Z^{BS})$, corresponding to the usual linear valuation rule \cite[Theorem 1.1]{kpq1997} without IM requirement. When the model is a one-dimensional geometric Brownian motion and $\xi=(S_T-K)_+$, the solution is given by the usual Black-Scholes formula.
\item The second simplest equation is $(V^{L},Z^{L})$ where the IM cost is computed using the ``Delta'' of an exogenous reference price given by the simpler pricing rule $(V^{BS},Z^{BS})$ without IM.
This is still a linear BSDE, but the numerical simulation requires to know $Z^{BS}$ in order to evaluate $(V^{L},Z^{L})$. We will use a nested Monte-Carlo procedure in our experiments.
\item The third equation is \eqref{eq:NLBSDE}, where the IM cost is computed using the ``Delta'' of the endogenous price $(V^{NL},Z^{NL})$ itself.
\end{itemize}

Existence and uniqueness of a solution to the BSDEs (\ref{eq:NLBSDE}), (\ref{eq:BSBSDE})
and (\ref{eq:LBSDE}) are direct consequences of \cite{pp1990}, as the
respective drivers satisfy standard Lipschitz properties and the processes $r$ and $\sigma^{-1}(r\1-\mu)$ are bounded.
\begin{proposition}\label{prop:BSDE:standard}
The BSDEs (\ref{eq:NLBSDE}), (\ref{eq:BSBSDE}) and (\ref{eq:LBSDE}) have a unique solution in the $L_2$-space $\bS_{0,T}^{2}\times 
\bH_{0,T}^{2}$, whose norm is uniformly bounded in $\Delta\leq T$.
\end{proposition}

The main result of this section is the following theorem.
\begin{thm}\label{theorem: approximation:price}
Define the $L_2$ time-regularity index of $Z^{NL}$ by
\begin{equation}
\label{eq:index:regularity:1}
{\cal E}^{NL}(\Delta):=\frac{ 1 }{ \Delta}\Esp{\int_{0}^{T}\int_{t}^{\left(t+\Delta\right)\wedge T}\left|Z_{s}^{NL}-Z_{t}^{NL}\right|^{2}\ds\dt}.
\end{equation}
We always have $\sup_{0<\Delta\leq T}{\cal E}^{NL}(\Delta)<+\infty$. Moreover,
there exist constants $K_{1},K_{2},K_{3}>0$, independent from $\Delta$,
such that
\begin{align}
||V^{L}-V^{BS}||_{\bS_{0,T}^{2}}^{2}+||Z^{L}-Z^{BS}||_{\bH_{0,T}^{2}}^{2} & \leq  K_{1}\Delta,
\label{eq:L-BS}
\\
||V^{NL}-V^{L}||_{\bS_{0,T}^{2}}^{2}+||Z^{NL}-Z^{L}||_{\bH_{0,T}^{2}}^{2} & \leq  K_{2}\Delta^{2}, 
\label{eq:NL-L}
\\
||V-V^{NL}||_{\bS_{0,T}^{2}}^{2}+||Z-Z^{NL}||_{\bH_{0,T}^{2}}^{2} & \leq  K_{3}\Delta\left(\Delta+ {\cal E}^{NL}(\Delta)\right).\label{eq:MV-NL}
\end{align}
In addition, we have 
\begin{equation}
\label{eq:estimate:fine:ENL}
{\cal E}^{NL}(\Delta)=O(\Delta),
\end{equation}
and thus $||V-V^{NL}||_{\bS_{0,T}^{2}}^{2}+||Z-Z^{NL}||_{\bH_{0,T}^{2}}^{2}=O(\Delta^2)$ provided that the additional sufficient conditions below are fulfilled:
\begin{enumerate}
\item [(i)] the terminal condition is a Lipschitz functional of $S$, that is,
$\xi=\Phi(S_{0:T})$ for some functional $\Phi$ satisfying 
$$|\Phi(x_{0:T})-\Phi(x'_{0:T})|\leq C_\Phi \sup_{t\in [0,T]}|x_t-x'_t|,$$
for any continuous paths $x,x':[0,T]\to\R^d$;
\item [(ii)] the coefficients $r,\sigma,\mu$ are constant.
\end{enumerate}
\end{thm}

\noindent Estimate \eqref{eq:estimate:fine:ENL} follows from \cite{zhang2004}. 
Let us remark that the tools used in \cite{zhang2004} to prove  \eqref{eq:estimate:fine:ENL}, and consequently the estimate \eqref{eq:estimate:fine:ENL} itself, should also hold if we assume (i) and the following more general conditions:
\begin{enumerate}
\item [(iii)] the processes
$r,\sigma,\mu$ are Markovian, i.e.\til$r_{t}=\hat{r}(t,S_{t})$,
$\sigma^{ij}_{t}=\hat \sigma^{ij}(t,S_{t})$ and
$\mu^i_{t}=\hat{\mu}^i(t,S_{t})$ for some deterministic functions 
$\hat{r},\hat \mu^i, \hat \sigma^{ij}$; 
\item [(iv)] the functions $x \to \hat{\mu }^i(x) x_i$, $x\to \hat{\sigma}^{ij}(x)x_i$ are globally Lipschitz in $(t,x)\in [0,T]\times \R^d$, for any $1\leq i ,j\leq d$;
\item [(v)] the functions $\hat{r}$ and $\hat \sigma^{-1}\left(\hat r\1-\hat \mu\right)$ are globally Lipschitz in $(t,x)\in[0,T]\times \R^d$.
\end{enumerate}

\noindent
\begin{remark}
When ${\cal E}^{NL}(\Delta)=O(\Delta)$, then $(V,Z)$ and $(V^{NL},Z^{NL})$ are close to each other: precisely, the squared error is of order $\Delta^2$.
Overall, the error estimates \eqref{eq:L-BS}-\eqref{eq:NL-L}-\eqref{eq:MV-NL} show that, on the one hand, there is a significative difference between valuation with or without initial margin cost: the size of the corrections $||V^{L}-V^{BS}||_{\bS_{0,T}^{2}}$ and $||Z^{L}-Z^{BS}||_{\bH_{0,T}^{2}}$ is of order $\sqrt \Delta$ (in \eqref{eq:L-BS}).
On the other hand, the other valuation rules yield comparable values when $\Delta\ll 1$, since the other corrections (\eqref{eq:NL-L} and \eqref{eq:MV-NL}) have size of order $\Delta$. 
\end{remark}

\subsection{Proof of Theorem \ref{theorem: approximation:price}}
\noindent \emph{$\rhd$ Estimate on ${\cal E}^{NL}(\Delta)$.} We start with a deterministic inequality. For any positive function $\Psi$ and any $\beta\geq 0$, we have
\begin{align}
\label{eq:det:inequality}
\int_0^T e^{\beta t}\left(\int_t^{(t+\Delta)\land T} \Psi_s \ds \right)\dt \leq 
\Delta \int_0^T e^{\beta s} \Psi_s \ds.
\end{align}
Indeed the left hand side of \eqref{eq:det:inequality} can be written as 
\begin{equation}
\label{eq:det:inequality:2}
\int_0^T\int_0^T e^{\beta t}\Psi_s\1_{t\leq s \leq (t+\Delta)\land T}\ds \dt
=\int_0^T\Psi_s \left(\int_0^T e^{\beta t}\1_{t\leq s \leq (t+\Delta)\land T}\dt\right) \ds,
\end{equation}
from which \eqref{eq:det:inequality} follows easily.
Using $(a+b)^2\leq 2 a^2+2b^2$ and \eqref{eq:det:inequality} with $\beta=0$ gives 
\begin{equation}
\label{eq:index:regularity:2}
{\cal E}^{NL}(\Delta)\leq
\frac{ 2 }{ \Delta}\Esp{\int_{0}^{T}\int_{t}^{\left(t+\Delta\right)\wedge T}
(|Z_{s}^{NL}|^{2}+|Z_{t}^{NL}|^{2})\ds\dt}
\leq 4 \, \Esp{\int_{0}^{T}\left|Z_{t}^{NL}\right|^{2} \dt},
\end{equation}
and the last term is uniformly bounded in $\Delta$ (Proposition \ref{prop:BSDE:standard}).

We now derive finer estimates that reveal the $L_2$ time-regularity of $Z^{NL}$ under the additional assumptions (i)-(ii). In this Markovian setting, we know that $Z^{NL}$ has a c\`adl\`ag version (see \cite[Remark (ii) after Lemma 2.5]{zhang2004}).
Then, introduce the time grid $t_{i}=i\Delta$ for $0\leq i\leq n:=\lfloor\frac{T}{\Delta}\rfloor$ and $t_{n+1}=T$. We claim that 
\begin{equation}
{\cal E}^{NL}(\Delta)\leq 4 \sum_{i=0}^{n}\Esp{\int_{t_{i}}^{t_{i+1}}\left|Z_{s}^{NL}-Z_{t_{i}}^{NL}\right|^{2}+\left|Z_{s}^{NL}-Z_{t_{i+1}}^{NL}\right|^{2}\ds}.\label{eq:decoupe}
\end{equation}
With this result at hand, the sestimate \eqref{eq:estimate:fine:ENL} follows directly from an application of \cite[Theorem 3.1]{zhang2004}. 
In sorder to show \eqref{eq:decoupe}, denote $\varphi^-(s)$ and $\varphi^+(s)$ the points on the grid before and after $s$, and write
\begin{align*}
\int_{0}^{T}\int_{t}^{\left(t+\Delta\right)\wedge T}|Z^{NL}_{s}&-Z^{NL}_{t}|^{2}\ds\dt
\leq 2\sum_{i=0}^{n}\int_{t_{i}}^{t_{i+1}}\int_{t}^{(t+\Delta)\land T}(
|Z^{NL}_{s}-Z^{NL}_{t_{i+1}}|^2+|Z^{NL}_{t}-Z^{NL}_{t_{i+1}}|^2)\ds\dt\\
&\leq 2\int_{0}^{T}\int_{t}^{\left(t+\Delta\right)\wedge T}\left(|Z^{NL}_{s}-Z^{NL}_{\varphi^-(s)}|^{2}+|Z^{NL}_{s}-Z^{NL}_{\varphi^+(s)}|^{2}\right)\ds\dt+2\Delta \int_{0}^{T}|Z^{NL}_{t}-Z^{NL}_{\varphi^+(t)}|^{2}\dt\\
&\leq 4\Delta\left( \int_{0}^{T}|Z^{NL}_{t}-Z^{NL}_{\varphi^-(t)}|^{2}\dt
+\int_{0}^{T}|Z^{NL}_{t}-Z^{NL}_{\varphi^+(t)}|^{2}\dt\right)
\end{align*}
where we have used \eqref{eq:det:inequality} with $\beta=0$.
The inequality \eqref{eq:decoupe} follows. 

\noindent \emph{$\rhd$ Proof of \eqref{eq:L-BS}.}
This  error estimate is related to the difference of two
linear BSDEs. The drivers of $V^{BS}$ and $V^{L}$ are respectively
$f^{BS}\left(s,y,z,\lambda\right)=-r_sy+z\sigma_{s}^{-1}\left(r_{s}\1-\mu_{s}\right)$,
and $$f^{L}\left(s,y,z,\lambda\right)=-r_sy+z\sigma_{s}^{-1}\left(r_{s}\1-\mu_{s}\right)+R C_{\alpha}\sqrt{\left(s+\Delta\right)\wedge T - s}|Z_{s}^{BS}|,$$
for $s\in [0,T]$, $v,\lambda \in\R$ and $z\in\R^{d}$,
hence
\[
\left(f^{L}-f^{BS}\right)\left(s,y,z,\lambda\right)=R C_{\alpha}\sqrt{(s+\Delta)\wedge T - s}|Z_{s}^{BS}|.
\]
By Lemma \ref{lem:NEKestimates}, we obtain that for $\mu>0$, $\beta$ large
enough and $K_{1}=\frac{C}{\mu^{2}}\left(R C_{\alpha}\right)^{2}||Z^{BS}||_{\bH_{\beta,T}^{2}}$,
\begin{eqnarray*}
||V^{L}-V^{BS}||_{\bS_{0,T}^{2}}^{2}+||Z^{L}-Z^{BS}||_{\bH_{0,T}^{2}}^{2} & \leq & ||V^{L}-V^{BS}||_{\bS_{\beta,T}^{2}}^{2}+||Z^{L}-Z^{BS}||_{\bH_{\beta,T}^{2}}^{2}\leq K_{1}\Delta.
\end{eqnarray*}
We are done with \eqref{eq:L-BS}.\medskip

\noindent \emph{$\rhd$ Proof of \eqref{eq:NL-L}.}
Then, as $\xi\in L^{2}$, as the processes $r,\sigma^{-1}\left(\mu-r\1\right)$
are bounded and as the non-linear term 
\[
t,z\in[0,T]\times\R^d\rightarrow \R C_{\alpha}\sqrt{\left(t+\Delta\right)\wedge T - t}|z|
\]
is Lipschitz in the variable $z$, uniformly in time, we obtain Inequality
(\ref{eq:NL-L}) as an application of \cite[Theorem 2.4]{GobetPagliarani}, for which assumptions $\mathbf{H.1-H.3}$ are satisfied.\medskip

\noindent \emph{$\rhd$ Proof of \eqref{eq:MV-NL}.}
Using computations similar
to those in the proof of Lemma \ref{lem:NEKestimates}, we obtain existence of a constant $C>0$ such that for $\mu>0$ and $\beta$ large enough,
\[
||V-V^{NL}||_{\bS_{\beta,T}^{2}}^{2}+||Z-Z^{NL}||_{\bH_{\beta,T}^{2}}^{2}\leq\frac{C}{\mu^2}\left|\left|\CVaRa_{\cF_{\cdot}}\left(V^{NL}_{\cdot}-V^{NL}_{\left(\cdot+\Delta\right)\wedge T}\right)-\CVaRa_{\cF_{\cdot}}\left(-\int_{\cdot}^{\left(\cdot+\Delta\right)\wedge T}Z_{\cdot}^{NL}\dW_{s}\right)\right|\right|_{\bH_{\beta,T}^{2}}^{2}.
\]
As the CVaR function is subadditive \cite{rock:urya:00}, we have that given $A,B$ two random variables,
and $t\in[0,T]$, $\CVaRaFt\left(A\right)\leq\CVaRaFt\left(B\right)+\CVaRaFt\left(A-B\right)$.
Inverting the roles of $A$ and $B$, we obtain that 
\begin{eqnarray*}
0\leq\left|\CVaRaFt\left(A\right)-\CVaRaFt\left(B\right)\right| & \leq & \max\left(\CVaRaFt\left(A-B\right),\CVaRaFt\left(B-A\right)\right)\\
 & \leq & \frac{1}{1-\alpha}\left(\Esp{\left(A-B\right)^{+}\mid\cF_{t}}+\Esp{\left(B-A\right)^{+}\mid\cF_{t}}\right)= \frac{\Esp{|A-B|\mid\cF_{t}}}{1-\alpha},
\end{eqnarray*}
where for the last inequality, we have used  that for $U\in\{A-B,B-A\}$,
$\inf_{x\in\R}\Esp{\frac{\left(U-x\right)^{+}}{1-\alpha}+x\mid\cF_{t}}\leq\Esp{\frac{U^{+}}{1-\alpha}\mid\cF_{t}}$.
We then have that 
\[
\left|\CVaRaFt\left(A\right)-\CVaRaFt\left(B\right)\right|^{2}\leq\frac{1}{\left(1-\alpha\right)^{2}}\Esp{\left(A-B\right)^{2}\mid \cF_{t}}.
\]
Setting, for $t\in[0,T]$, $A_{t}=-\int_{t}^{(t+\Delta)\land T}Z_{t}^{NL}\dW_{s}$, $B_{t}=V_{t}^{NL}-V_{\left(t+\Delta\right)\wedge T}^{NL}$ and
using the previous inequality, we obtain
\begin{align*}
||V-V^{NL}||_{\bS_{\beta,T}^{2}}^{2}+||Z-Z^{NL}||_{\bH_{\beta,T}^{2}}^{2} & \leq  \frac{C}{\mu^2\left(1-\alpha\right)^{2}}\Esp{\int_{0}^{T}e^{\beta t}\left(V^{NL}_{t}-V^{NL}_{\left(t+\Delta\right)\wedge T}+\int_{t}^{(t+\Delta)\land T}Z_{t}^{NL}\dW_{s}\right)^{2}\dt}.
\end{align*}
Now writing
\begin{align*}
V^{NL}_{t}-V^{NL}_{\left(t+\Delta\right)\wedge T}&+\int_{t}^{\left(t+\Delta\right)\wedge T}Z_{t}^{NL}\dW_{s} \\
& =  \int_{t}^{\left(t+\Delta\right)\wedge T}\left(-r_{s}V_{s}^{NL}+Z_{s}^{NL}\sigma_{s}^{-1}\left(r_{s}\1-\mu_{s}\right)
+ R \, C_{\alpha}\sqrt{\left(s+\Delta\right)\wedge T-s}|Z_{s}^{NL}|\right)\ds\\
 & \quad -  \int_{t}^{\left(t+\Delta\right)\wedge T}\left(Z_{s}^{NL}-Z_{t}^{NL}\right)\dW_{s}=:\Pi_{1}(t)-\Pi_{2}(t),
\end{align*}
we obtain $||V-V^{NL}||_{\bS_{\beta,T}^{2}}^{2}+||Z-Z^{NL}||_{\bH_{\beta,T}^{2}}^{2}\leq\frac{2C}{\mu^2\left(1-\alpha\right)^{2}}\Esp{\int_{0}^{T}e^{\beta t}\left(\Pi^2_{1}(t)+\Pi^2_{2}(t)\right)\dt}.$
By Jensen's inequality and the inequality \eqref{eq:det:inequality}, we get
\begin{align*}
\Esp{\int_{0}^{T}e^{\beta t}\Pi^2_{1}(t)\dt}&\leq3\Delta\Esp{\int_{0}^{T}e^{\beta t}\int_{t}^{\left(t+\Delta\right)\wedge T}\left(\left(-r_{s}V_{s}^{NL}\right)^{2} + \left(Z_{s}^{NL}\sigma_{s}^{-1}\left(r_{s}\1-\mu_{s}\right)\right)^{2} + \left(R \, C_{\alpha}\sqrt{\Delta}|Z_{s}^{NL}|\right)^{2}\right)\ds\dt}\\
&\leq3\Delta^2 \left(|r|_\infty^2+ |\sigma^{-1}(r\1-\mu)|_\infty^2 + \left(R \, C_{\alpha}\right)^2\right)\Esp{\int_{0}^{T}e^{\beta t}(|V_{t}^{NL}|^{2}+|Z_{t}^{NL}|^2)\dt}.
\end{align*}
Now invoking the uniform estimate in Proposition \ref{prop:BSDE:standard}, we finally obtain 
$\Esp{\int_{0}^{T}e^{\beta t}\left(\Pi_{1}\left(t\right)\right)^{2}\dt}\leq \tilde{K_{1}}\Delta^{2}$
for some positive constant $\tilde{K}_{1}$.
Moreover, using Ito's isometry, we obtain
\[
\Esp{\int_{0}^{T}e^{\beta t}\Pi^2_{2}(t)\dt}=\Esp{\int_{0}^{T}e^{\beta t}\int_{t}^{\left(t+\Delta\right)\wedge T}\left|Z_{s}^{NL}-Z_{t}^{NL}\right|^{2}\ds\dt}\leq e^{\beta T}\Delta \, {\cal E}^{NL}(\Delta).\]
Gathering all the previous arguments, we get \eqref{eq:MV-NL}.
The proof of Theorem \ref{theorem: approximation:price} is now complete.
\qed

\section{Numerical Examples}
\label{section:Numerical Examples}
In the absence of numerical methods to estimate the solution of the McKean Anticipative BSDE \eqref{eq:CvarBSDE} in full generality, we rather solve numerically the BSDE approximations \eqref{eq:NLBSDE} or \eqref{eq:LBSDE} as discussed in Section \ref{sec:approximationsBSDE}.
For this purpose, when the dimension $d$ is greater than one, we use the Stratified Regression Multistep-forward Dynamical Programming (SRMDP) scheme developed in \cite{gstv2016}. In our numerical tests in this section, we set the coefficients $\mu_t, \sigma_t$ and $r_t$ of the model \eqref{eq:ito market} to be constant (multi-dimensional geometric Brownian motion).
\\
\red{Note that the linear BSDE \eqref{eq:LBSDE} is solved by $V^L_t = v^L(t,S_t)$, where $S$ follows \eqref{eq:ito market} and $v$ is the solution to the (linear) PDE
\begin{equation} \label{e:linear_PDE}
\partial_t v^L(t,x) +
\bigl(r x \partial_x +\frac 12 \sigma^2 x^2 \partial_{xx} \bigr) v^L(t,x) + R \, C_{\alpha} \sqrt{\left(t +\Delta\right)\wedge T - t} \, \sigma |\delta^{BS}(t,x)| x - r \, v^L(t,x) = 0,
\end{equation}
and $v^L(T,x) = \Phi(x)$, where $\delta^{BS}$ is the option's delta in the Black-Scholes model with interest rate $r$ and volatility parameter $\sigma$.
Section \ref{s:NestedMC} is precisely devoted to the numerical solution of this PDE via an appropriate Nested Monte Carlo method.
In particular, the initial value $V^L_0 = v^L(0, S_0)$ does not depend on the drift parameter $\mu$.
The same comment applies to the non-linear BSDE \eqref{eq:NLBSDE}: this equation is solved by $V^{NL}_t = v^{NL}(t,S_t)$, where $v^{NL}$ is the solution of the (non-linear) PDE \eqref{eq:fdeq}-\eqref{eq:fdtc}, which, once again, does not depend on $\mu$.
When the initial portfolio values $V^L_0$ or $V^{NL}_0$ are evaluated via a Monte Carlo method that requires to simulate the underlying process $S$, the error of the algorithm will, in general, display a dependence with respect to the drift parameter $\mu$.
For simplicity, we take $\mu^i = r$ in our experiments.}
\\
Observe that setting $R=0$ reduces the original BSDE to the linear equation \eqref{eq:BSBSDE}.
The (explicit) solution in this case will serve as a reference value in order to assess the impact of Initial Margins.

\subsection{Finite difference method for $(V^{NL},Z^{NL})$ in dimension 1} \label{s:fin_diff}

In order to check the validity of our results, we first obtain a benchmark when $d=1$ by solving the semi-linear parabolic PDE related to the BSDE \eqref{eq:NLBSDE} when $\xi = \Phi(S_T)$, see \cite{pard:rasc:14}. By an application of  It\^{o}'s lemma, the semi-linear PDE is given by
\begin{align}
\label{eq:fdeq}
& \dfrac{\partial V}{\partial t} + \dfrac{1}{2}\sigma^2 S^2 \dfrac{\partial^2 V}{\partial S^2} + r S \dfrac{\partial V}{\partial S} + C_\alpha R \, \sigma \sqrt{(t+\Delta)\wedge T - t} \bigg| S \dfrac{\partial V}{\partial S}\bigg| - r\,V = 0, \quad  (t,S) \in [0,T) \times  \R^+, \\
\label{eq:fdtc}
& V(T,S) = \Phi(S),\quad S \in \R^+,
\end{align}
and $(V^{NL}_t,Z^{NL}_t)=(V(t,S_t),\frac{\partial V}{\partial S}(t,S_t)\sigma S_t)$.

\begin{remark}\label{remark:BS}
If $\Phi(S) = \max(S-K,0)$ or $\max(K-S,0)$ for some $K > 0,$ i.e., either a call or a put option payoff, we expect the gradient $\tfrac{\partial V}{\partial S}$
to have a constant sign.
In such a case, the PDE \eqref{eq:fdeq}--\eqref{eq:fdtc} becomes linear and in fact has an explicit solution, given by a Black-Scholes formula with time-dependent continuous dividend yield $d(t) = - C_{\alpha} R \, \sigma \sqrt{(t+\Delta) \wedge T - t} \ \mathrm{sign}(\tfrac{\partial V}{\partial S})$.
\end{remark}

We use a classical finite difference methods to solve \eqref{eq:fdeq}-\eqref{eq:fdtc} (see, for example, \cite{achd:piro:05}). First, we perform a change of variable, $x = \ln S$, so that the PDE can be rewritten in the following form for the function $v(t,x) := V(t,e^x)$:
\begin{align}
\label{eq:Xfdeq}
& \dfrac{\partial v}{\partial t} + \dfrac{1}{2}\sigma^2\dfrac{\partial^2 v}{\partial x^2} + \Bigl(r-\frac{1}{2}\sigma^2\Bigr) \dfrac{\partial v}{\partial x} + C_\alpha R \sigma \sqrt{(t+\Delta)\wedge T - t} \bigg| \dfrac{\partial v}{\partial x}\bigg| - r\,v = 0, \quad  (t,x) \in [0,T) \times  \R,\\
\label{eq:Xfdtc}
& v(T,x) = \Phi({e^x}), \quad  x \in \R.
\end{align}
We denote the finite difference domain by $D = [0,T] \times [x_{\min},x_{\max}]$ with $-\infty < x_{\min} < x_{\max} < \infty$. The domain $D$ is approximated with a uniform mesh $\mathcal{D} = \bigl\{(t^n,x_i): n = 0,1, \ldots,N, \, i = 0,1,\ldots,M \bigr\}$, where $t^n := n \Delta t$ and $x_i := x_{\min} + i \Delta x.$ Here, for $N$ time intervals, $\Delta t = T/N$ and $\Delta x = (x_{\max}-x_{\min})/M$ for $M$ spatial steps. Furthermore, we denote $v(t^n,x_i) = v_i^n$. Next, consider the following finite difference derivative approximations under the well-known $\omega$-scheme, i.e., we replace $v_i^n$ by $\omega v_i^n + (1-\omega)v_i^{n+1}$, where $\omega \in [0,1]$ is a constant parameter, such that
\eqstar{
\dfrac{\partial v}{\partial t}(t^n,x_i) &\approx \dfrac{v_i^{n+1} - v_i^n}{\Delta t}, \quad \dfrac{\partial v}{\partial x}(t_n,x_i) \approx \omega \dfrac{v_{i+1}^n - v_i^n}{\Delta x} + (1-\omega) \dfrac{v_{i+1}^{n+1} - v_i^{n+1}}{\Delta x},\\
 \dfrac{\partial^2 v}{\partial x^2}(t_n,x_i) &\approx \omega \dfrac{v_{i+1}^n - 2 v_i^n + v_{i-1}^n}{(\Delta x)^2} + (1-\omega) \dfrac{v_{i+1}^{n+1} - 2 v_i^{n+1} + v_{i-1}^{n+1}}{(\Delta x)^2}.
}
The choice $\omega = 0.5$ corresponds to Crank-Nicolson method.  We also ``linearize'' the non-linear term by treating it as explicit, i.e., at any time $t_n$ we take $\bigg | \dfrac{\partial v}{\partial x}(t^n,x_i) \bigg | \approx \bigg | \dfrac{\partial v}{\partial x}(t^{n+1},x_i) \bigg |.$

The substitution of finite difference derivative approximations in \eqref{eq:Xfdeq}-\eqref{eq:Xfdtc} along with the ``linearization'' step, leads to the following tridiagonal linear system at each time step $n=N-1,\ldots,0$ which can be solved by Thomas algorithm \cite{yg1973}: $A v^n = b^{n+1},$ with nonzero coefficients of the tridiagonal matrix $A = (a_{i,j})$ given by
\eqstar{
 a_{0,0} &= 1, \quad & a_{0,1} &= 0, &  \quad  & a_{M,M-1} = 0,  \quad  a_{M,M} = 1, \\
 a_{i,i}  &= 1 + 2\theta  \omega + \kappa \omega + \rho  \omega,  \quad & a_{i,i+1}  &= -\theta \omega - \kappa \omega, & \quad & a_{i-1,i}  = -\theta \omega, \quad i=1,\ldots,M-1, 
}
and the time dependent vector $b^{n+1}$ as:
\eqstar{
 b_0^{n+1} = v^{n+1}_{0}, \quad  b_{M}^{n+1}  = v^{n+1}_{M}, \quad b_i^{n+1} &= \theta (1-\omega) v_{i-1}^{n+1} + (1-2\theta (1-\omega) - \kappa(1-\omega) - \rho (1-\omega) ) v_i^{n+1} \\
&\quad+ (\theta  (1-\omega) + \kappa(1-\omega) ) v_{i+1}^{n+1} + \beta^{n} | v_{i+1}^{n+1} - v_{i}^{n+1}|, \quad i=1,\ldots,M-1,
}
where $v^{n+1}_{0}$, $v^{n+1}_{M}$ are given by the boundary conditions and the remaining constants are defined as below
\eqstar{
 \theta = \dfrac{ \sigma^2\Delta t}{2 (\Delta x)^2},\quad  \kappa = \dfrac{(r-\frac{ 1 }{2 }\sigma^2)\Delta t}{\Delta x}, \quad 
 \rho = r \Delta t, \quad  \beta^n = \dfrac{C_\alpha R \sigma \Delta t \sqrt{(t_n + \Delta)\wedge T - t_n}}{\Delta x}. 
}
The $i$th coordinate of vector $v^n$ is the approximation of the value $v(t^n,x_i)$. 

We set the model parameters as $T = 1$, $\sigma = 0.25$, $r   = 0.02$, $\alpha=0.99$ and $\Delta = 0.02$ (1 week) and consider three different options -- call, put and butterfly, for different strikes. We set $R=0.02$ when accounting for IM and $R=0$ otherwise. The space domain we consider for the finite difference scheme is $[\ln(10^{-6}),\ln(4 K)]$, while for the SRMDP regression algorithm, we fix the space domain to be $[-5,5]$. Furthermore, for the finite difference scheme, the number of steps are $N=10^3$ and $M=10^6$.
For the \textbf{LP0} version of the SRMDP algorithm, we consider $2800$ hypercubes, $50$ time steps, and $2500$ simulations per hypercube.
In Figure \ref{fig:callNL} and Table \ref{tab:callNL}, we present the results for implied volatilities, prices and deltas of call options with different strikes. First, we compute the values using the classical Black-Scholes formula (the curve labeled B-S $R=0$) in order to assess the impact of taking into account the IM.
Next, we solve the non-linear BSDE \eqref{eq:NLBSDE} using the three methods presented above: the exact Black-Scholes formula when IM is considered as a time-dependent dividend yield (labeled BS $R=0.02$) (see Remark \ref{remark:BS}), the finite difference method (FD) and the SRMDP algorithm. For the last method, we compute $95\%$ confidence intervals for the price and the delta of the options. 
The results for put options are given in Figure \ref{fig:putNL} and Table \ref{tab:putNL}.
In all cases, we observe that the IM has a significant impact on the implied volatility of option prices (around 20-30 basis points for standard values of the volatility). 
Finally, we consider butterfly options with payoff function 
\begin{equation} \label{e:butterfly_payoff}
\Phi(S_T) = \left(S_T - (K-2) \right)^{+} - 2 \left( S_T - K \right)^{+} + \left(S_T - (K+2) \right)^{+}
\end{equation}
in Figure \ref{fig:butterflyNL} and Table \ref{tab:butterflyNL}.
This derivative product involves three options with different strikes -- a long position in the butterfly consists in  buying a call option with strike $K-2$, a call with strike $K+2$, and selling two call options with strike $K$.
Note that the first derivative of the price of a butterfly option (the option delta) is expected to take values of both signs, therefore an explicit Black-Scholes formula is not available any more when IM is taken into account the ($R=0.02$ here).
For all the payoffs above, we observe that the SRMDP algorithm provides good accuracy when compared to the exact values or finite difference estimates.

\red{As a side remark, we note that, due to the non-linear pricing rule $V^{NL}$, the call and put prices that we compute do not satisfy the put-call parity relation. As a consequence, the implied volatilities extracted from these call and put prices do not coincide, as Figures \ref{fig:callNL} and \ref{fig:putNL} show. 
We analyse more in detail this situation in the appendix. 
In particular, we explain why, in this context, the implied volatility smile is decreasing with respect to  strike for calls and increasing for puts.}
\medskip

\begin{figure}[!htb]
  \centering
  \includegraphics[scale=0.5]{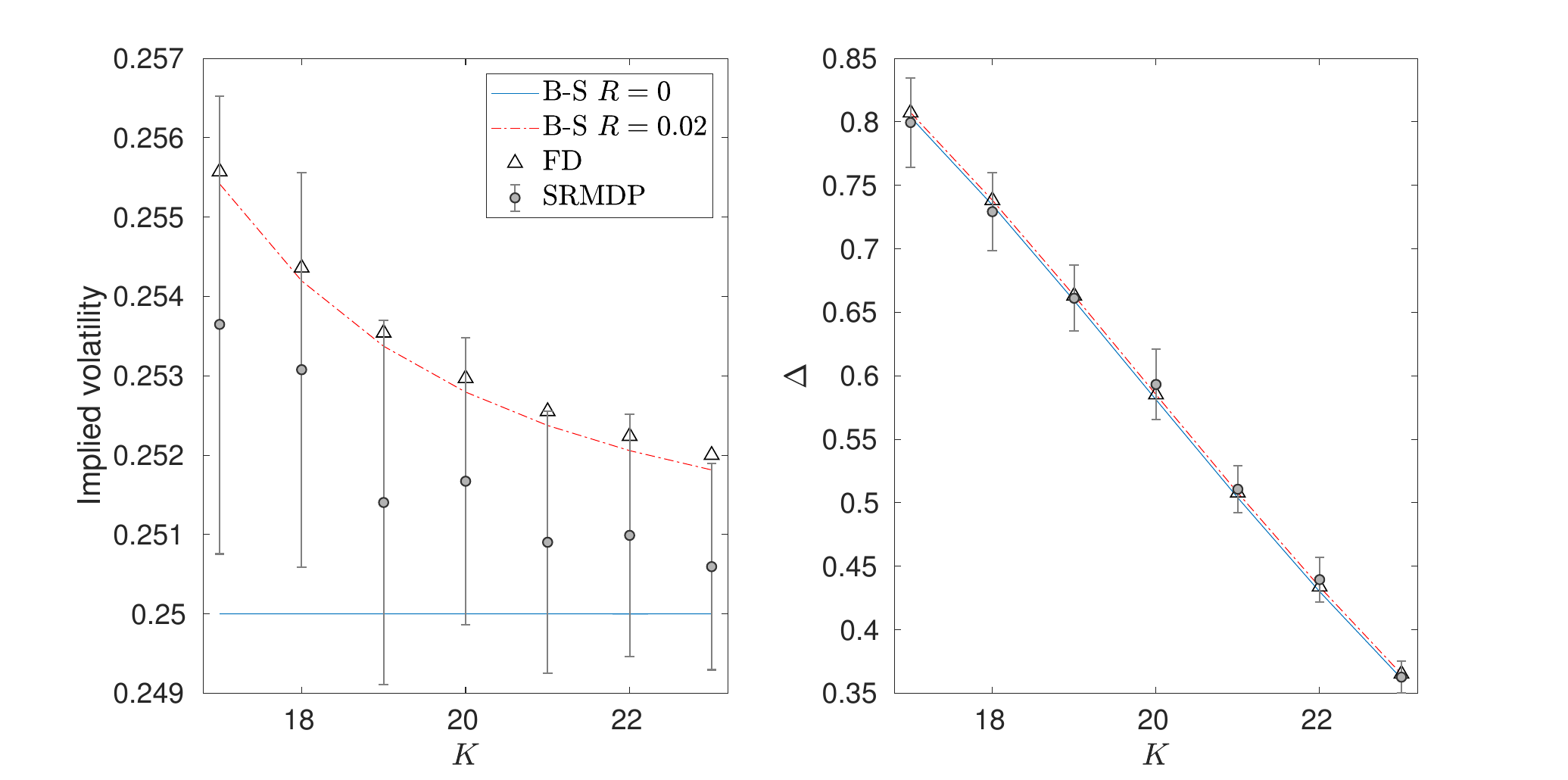}\par
  \caption{Implied volatility and delta for call options with spot value $S_0=20$ and different strikes $K$}
  \label{fig:callNL}
\end{figure}

\begin{table}[htb]
\centering
\small
\begin{tabular}{c |c c c c c c c c}
\toprule
 & \multicolumn{2}{c}{B-S $R=0$} & \multicolumn{2}{c}{B-S $R=0.02$} & \multicolumn{2}{c}{FD} & \multicolumn{2}{c}{SRMDP}\\
  \cline{2-9} 
&&&&&&& 95\% CI  & 95\% CI \\
$K$ & B-S$(0,S_0)$ & $\Delta(0,S_0)$ & B-S$(0,S_0)$ & $\Delta(0,S_0)$ & $V(0,S_0)$ & $\nabla V(0,S_0)$ &  $V_0^{NL}(S_0)$ &  $Z_0^{NL}(S_0)/(\sigma S_0)$  \\
\toprule
$17$ & $3.9534$ & $0.8037$ & $3.9835$ & $0.8073$ & $3.9844$ & $0.8072$ & $[3.9575,3.9897]$ & $[0.7641,0.8347]$ \\
 \hline
$18$ & $3.2795$ & $0.7345$ & $3.3071$ & $0.7383$ & $3.3082$ & $0.7382$ & $[3.2833,3.3161]$ & $[0.6986,0.7598]$ \\
 \hline
$19$ & $2.6863$ & $0.6592$ & $2.7111$ & $0.6631$ & $2.7123$ & $0.6630$ & $[2.6797,2.7134]$ & $ [0.6350,0.6871]$ \\
 \hline
$20$ & $2.1741$ & $0.5812$ & $2.1959$ & $0.5852$ & $2.1973$ & $0.5852$ & $[2.1730,2.2012]$ & $[0.5656,0.6207]$ \\
 \hline
$21$ & $1.7398$ & $0.5039$ & $1.7587$ & $0.5079$ & $1.7601$ & $0.5078$ & $[1.7338,1.7601]$ & $[0.4920,0.5292]$ \\
 \hline
$22$ & $1.3777$ & $0.4301$ & $1.3939$ & $0.4338$ & $1.3953$ & $0.4338$ & $[1.3734,1.3975]$ & $[0.4218,0.4571]$ \\
 \hline
$23$ & $1.0805$ & $0.3617$ & $1.0941$ & $0.3651$ & $1.0954$ & $0.3652$ & $[1.0752,1.0946] $ & $[0.3503,0.3750]$ \\
 \hline  
 \end{tabular}
\vspace{0.25pc}
\caption{Price and delta for call options with spot value $S_0=20$ and different strikes $K$.}
\label{tab:callNL}
\end{table}

\begin{figure}[!htb]
  \centering
  \includegraphics[scale=0.5]{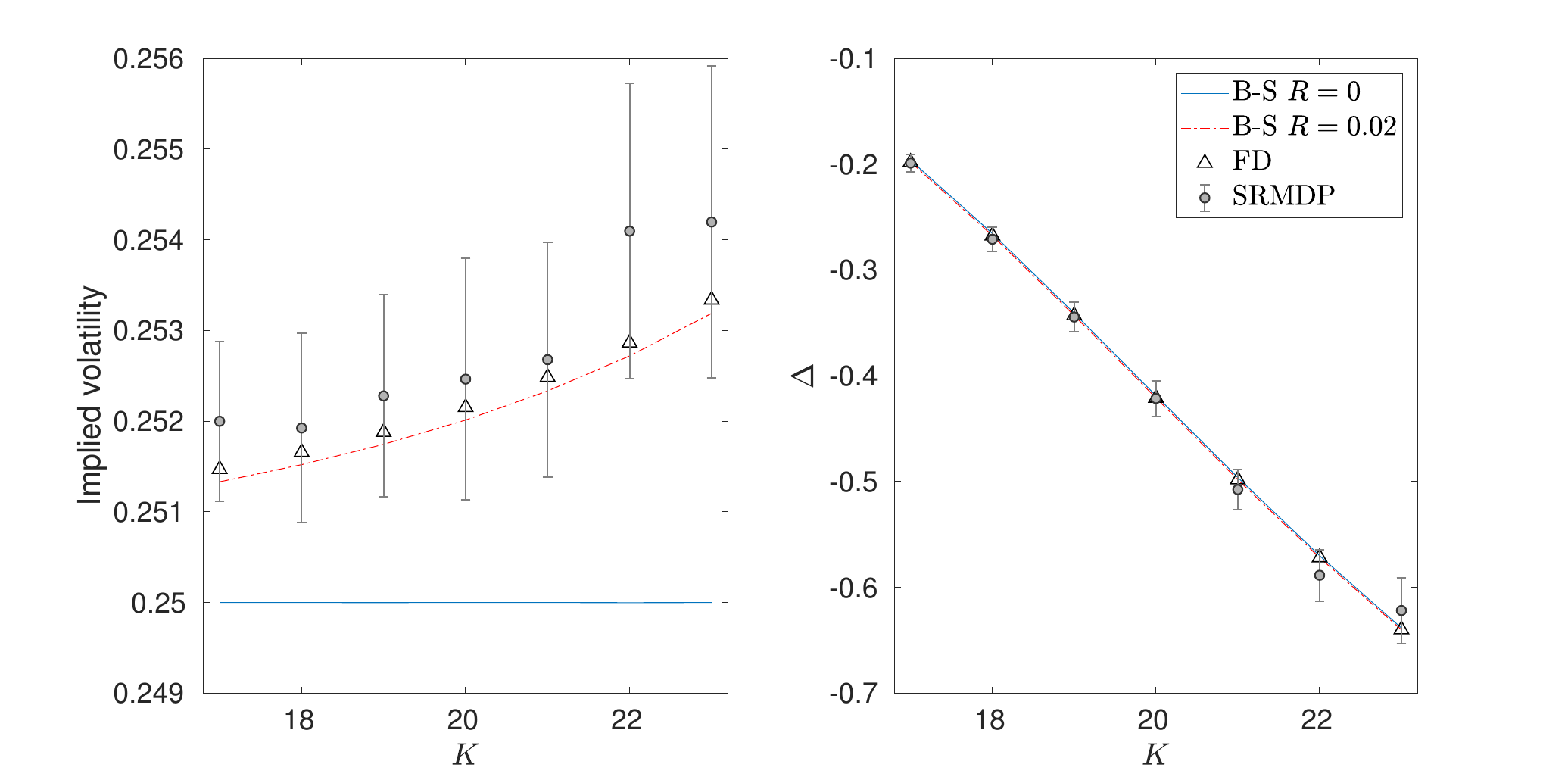}\par
  \caption{Implied volatility and delta for put options with spot value $S_0=20$ and different strikes $K.$}
  \label{fig:putNL}
\end{figure}

\begin{table}
\small
\centering
\begin{tabular}{c |c c c c c c c c}
\toprule
& \multicolumn{2}{c}{B-S $R=0$} & \multicolumn{2}{c}{B-S $R=0.02$} & \multicolumn{2}{c}{FD} & \multicolumn{2}{c}{SRMDP}\\
  \cline{2-9} 
&&&&&&& 95\% CI &95\% CI\\
$K$ & B-S$(0,S_0)$ & $\Delta(0,S_0)$ & B-S$(0,S_0)$ & $\Delta(0,S_0)$ & $V(0,S_0)$ & $\nabla V(0,S_0)$ &$V_0^{NL}(S_0)$ &  $Z_0^{NL}(S_0)/(\sigma S_0)$  \\
\toprule
$17$ & $0.6168$ & $-0.1963$ & $0.6241$ & $-0.1980$ & $0.6249$ & $-0.1981$ & $[0.6229,0.6328]$ & $[-0.2071,-0.1908]$ \\
 \hline
$18$ & $0.9231$ & $-0.2655$ & $0.9331$ & $-0.2675$ & $0.9340$ & $-0.2676$ & $[0.9289,0.9426]$ & $[-0.2827,-0.2592]$ \\
 \hline
$19$ & $1.3101$ & $-0.3408$ & $1.3229$ & $-0.3429$ & $1.3239$ & $-0.3430$ & $[1.3186,1.3350]$ & $[-0.3585,-0.3303]$ \\
 \hline
$20$ & $1.7781$ & $-0.4188$ & $1.7938$ & $-0.4209$ & $1.7949$ & $-0.4209$ & $[1.7869,1.8077]$ & $[-0.4383,-0.4046]$ \\
 \hline
$21$ & $2.3240$ & $-0.4961$ & $2.3426$ & $-0.4981$ & $2.3438$ & $-0.4981$ & $[2.3350,2.3557]$ & $[-0.5262,-0.4888]$ \\
 \hline
$22$ & $2.9421$ & $-0.5699$ & $2.9635$ & $-0.5718$ & $2.9646$ & $-0.5717$ & $[2.9615,2.9871]$ & $[-0.6128,-0.5641]$ \\
 \hline
$23$ & $3.6251$ & $-0.6383$ & $3.6490$ & $-0.6400$ & $3.6501$ & $-0.6398$ & $[3.6436,3.6695]$ & $[-0.6529,-0.5906]$ \\
 \hline  
\end{tabular}
\vspace{0.25pc}
\caption{Price and delta for put options with spot value $S_0=20$ and different strikes $K$.}
\label{tab:putNL}
\end{table}

\begin{figure}[!htb]
  \centering
  \includegraphics[scale=0.5]{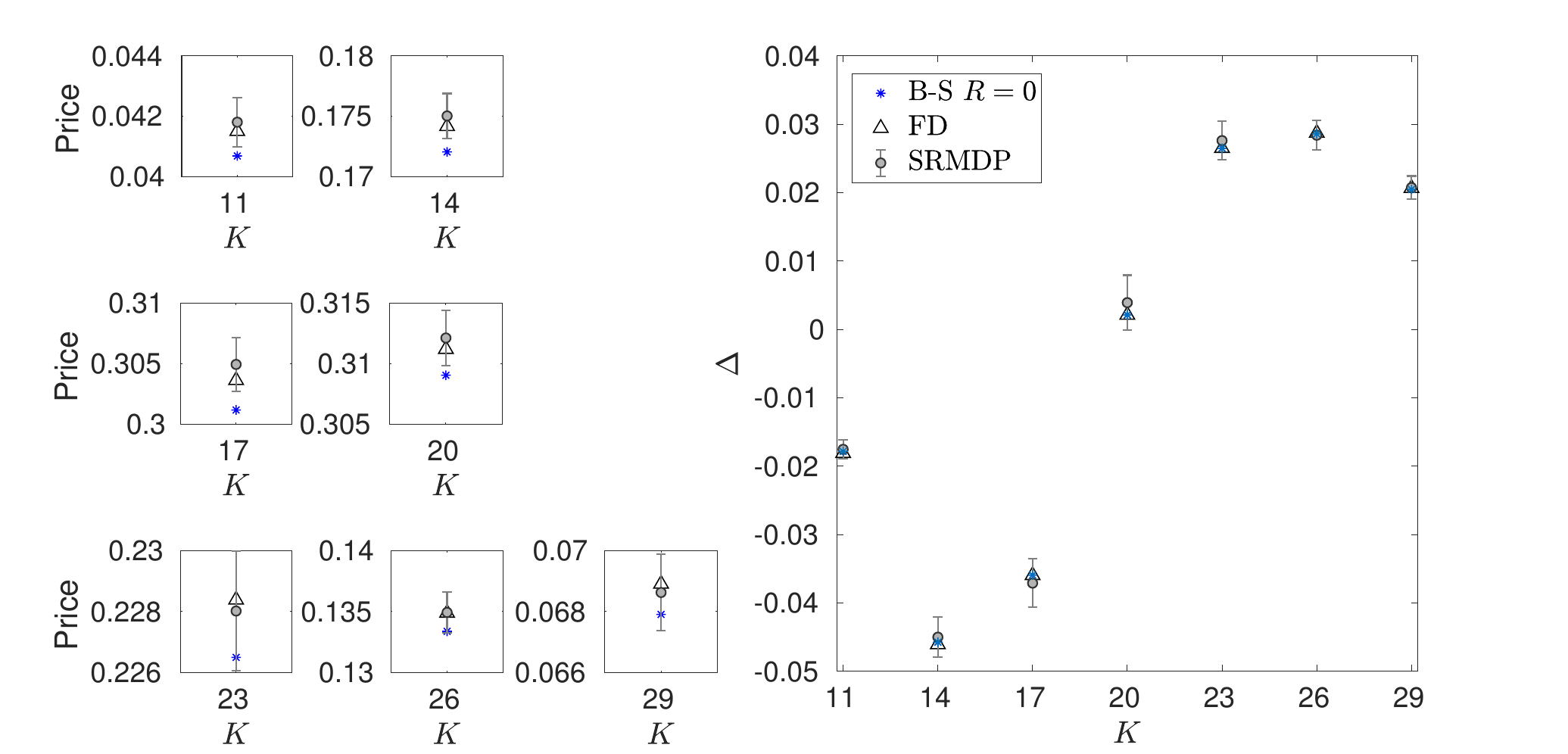}\par
  \caption{Price and delta for the butterfly option \eqref{e:butterfly_payoff} with spot value $S_0=20$ and different strikes $K$.}
  \label{fig:butterflyNL}
\end{figure}

\begin{table}
\small
\centering\begin{tabular}{c |c c c c c c}
\toprule
 & \multicolumn{2}{c}{B-S $R=0$} & \multicolumn{2}{c}{FD} & \multicolumn{2}{c}{SRMDP}\\
  \cline{2-7} 
  &&&&&95\% CI&95\% CI\\
$K$ & B-S$(0,S_0)$ & $\Delta(0,S_0)$ & $V(0,S_0)$ & $\nabla V(0,S_0)$ &  $V_0^{NL}(S_0)$ &  $Z_0^{NL}(S_0)/(\sigma S_0)$  \\
\toprule
$11$ & $0.0407$ & $-0.0178$ & $0.0415$ & $-0.0181$ & $[0.0410,0.0426]$ & $[-0.0189,-0.0162]$ \\
 \hline
$14$ & $0.1720$ & $-0.0457$ & $0.1742$ & $-0.0461$ & $[0.1732,0.1769]$ & $[-0.0479,-0.0421]$ \\
 \hline
$17$ & $0.3012$ & $-0.0359$ & $0.3036$ & $-0.0359$ & $[0.3027,0.3072]$ & $[-0.0406,-0.0336]$ \\
 \hline
$20$ & $0.3090$ & $\phantom{-}0.0021$ & $0.3112$ & $\phantom{-}0.0021$ & $[0.3098,0.3144]$ & $[-0.0001,\phantom{-}0.0079]$ \\
 \hline 
$23$ & $0.2265$ & $\phantom{-}0.0265$ & $0.2284$ & $\phantom{-}0.0265$ & $[0.2261,0.2300]$ & $[\phantom{-}0.0248,\phantom{-}0.0304]$ \\
 \hline 
$26$ & $0.1334$ & $\phantom{-}0.0286$ & $0.1349$ & $\phantom{-}0.0287$ & $[0.1333,0.1366]$ & $[\phantom{-}0.0263,\phantom{-}0.0305]$ \\
 \hline
$29$ & $0.0679$ & $\phantom{-}0.0205$ & $0.0689$ & $\phantom{-}0.0207$ & $[0.0674,0.0699]$ & $[\phantom{-}0.0190,\phantom{-}0.0224]$ \\
 \hline 
\end{tabular}
\vspace{0.25pc}
\caption{Price and delta for the butterfly option \eqref{e:butterfly_payoff} with spot value $S_0=20$ and different strikes $K$.}
\label{tab:butterflyNL}
\end{table}  

\subsection{Variance reduction for solving $(V^{NL},Z^{NL})$ using $(V^{BS},Z^{BS})$}
In order to asses the impact of using $R > 0$  on the solution of the BSDE \eqref{eq:NLBSDE}, in the case of European call and put options in one dimension, it is better to solve the BSDE difference  $(V_t^{DF},Z_t^{DF}) = (V_t^{NL} -V_t^{BS},Z_t^{NL} -Z_t^{BS})$ which has a reduced variance in the algorithm. Note that for a call option
$$Z_t^{BS} = \sigma S_t \varPhi(d_1), \quad d_{1,t} = \dfrac{\ln(S_t/K) + (r + \frac{1}{2}\sigma^2)(T-t)}{\sigma \sqrt{T-t}},$$
where $\varPhi$ is the standard Gaussian cumulative distribution function.
Therefore $|Z^{NL}_t| = |Z_t^{DF} + \sigma S_t\varPhi(d_{1,t})|$. Then, the BSDE for the difference $(V_t^{DF},Z_t^{DF})$ in the case of a call option\footnote{For a put option an analogous BSDE can be written taking into account that  $Z_t^{BS} = \sigma S_t (\varPhi(d_{1,t})-1)$.} is given by: $$ V_t^{DF}  =  0 + \int_t^T \left(-r V^{DF}_s + C_\alpha R \sqrt{(s+\Delta)\wedge T - s} |Z^{DF}_s + \sigma S_s \varPhi(d_{1,s})|\right) \ds -\int_t^T Z^{DF}_s \dW_s. $$

In Table \ref{tab:bsde_diff}, the BSDE $(V_t^{DF},Z_t^{DF})$ is solved for several call and put options using the SRMDP algorithm. Besides, exact solutions (ES) are computed through the difference between Black-Scholes formula where IM's contribution is considered as a time-dependent dividend yield and the classical Black-Scholes formula with $R=0.$

Once again these tests allow us to demonstrate that SRMDP algorithm provides accurate results in one dimension. 
\begin{table}[!htb]
\centering
\small
\begin{tabular}{c | c c c c }
\toprule
 & \multicolumn{2}{c}{ES: $\mbox{B-S } R=0.02 - \mbox{B-S } R=0$ } & \multicolumn{2}{c}{SRMDP}\\
  \cline{2-5} 
  &&&95\% CI&95\% CI\\
& B-S$(0,S_0)$ & $\Delta(0,S_0)$  & $V_0^{DF}(S_0)$ & $Z_0^{DF}(S_0)/(\sigma S_0)$  \\
\toprule
 Call, $K = 17$ & $0.0302$ & $\phantom{-}0.0036$ & $[0.0302,0.0304]$ & $[\phantom{-}0.0033,\phantom{-}0.0037]$  \\
 \hline
 Call, $K = 20$ & $0.0218$ & $\phantom{-}0.0040$ & $[0.0218,0.0219]$ & $[\phantom{-}0.0038,\phantom{-}0.0042]$  \\
 \hline
 Call, $K = 23$ & $0.0136$ & $\phantom{-}0.0035$ & $[0.0136,0.0137]$ & $[\phantom{-}0.0033,\phantom{-}0.0036]$  \\
 \hline
 Put, $K = 17$ & $0.0074$ & $-0.0017$ & $[0.0074,0.0075]$ & $[-0.0020,-0.0016]$  \\
 \hline
 Put, $K = 20$ & $0.0157$ & $-0.0021$ & $[0.0157,0.0158]$ & $[-0.0023,-0.0019]$  \\
 \hline
 Put, $K = 23$ & $0.0239$ & $-0.0016$ & $[0.0239,0.0240]$ & $[-0.0018,-0.0015]$  \\
 \hline  
\end{tabular}
\vspace{0.3cm}
\caption{SRMDP algorithm for the BSDE $(V_t^{DF},Z_t^{DF})$.}
\label{tab:bsde_diff}
\end{table}

\subsection{Nested Monte Carlo for computing $(V^L,Z^L)$ in dimension 1} \label{s:NestedMC}

As discussed in Section \ref{sec:approximationsBSDE}, we can further approximate the solution of the non-linear BSDE $V^{NL}$ by the linear BSDE $V^L$ with external source term $Z = Z^{BS}.$ In this case, we have an explicit representation for $Z_t^{BS}$, given by
\eqstar{
Z_t^{BS}  = \partial_s v^{BS}(t,S_t) \, \sigma S_t,
} 
where $v^{BS}(t,s) := \E[ e^{-r(T-t)} \Phi(S_T) | S_t = s ].$  We can use the likelihood ratio method of Broadie and Glasserman \cite{broadie1996estimating} (see also the automatic differentiation formula in \cite[Proposition 6.2.1]{nualart:06}) to rewrite the derivative in terms of an expectation, getting 
\eqlnostar{eq:formulaDelta}{
Z_t^{BS}(s) = \frac{\partial}{\partial s}  \E[ e^{-r(T-t)} \Phi (S_T) | S_t =s ] \sigma s 
&= \E\left[ e^{-r(T-t)} \Phi(S_T) \frac{W_T - W_t}{\sigma S_t (T-t)}  \Big | S_t =s \right] \sigma s \nonumber
\\
&= \E\left[ e^{-r(T-t)} \bigl(\Phi(S_T) - \Phi(S_t) \bigr) \frac{W_T - W_t}{T-t} \Big | S_t =s \right].
}
The additional conditionally centered term $-\Phi(S_t) \frac{W_T - W_t}{T-t}$ that we have introduced in the last expression plays the role of a control variate, allowing to reduce the variance in the Monte-Carlo simulation.
Recall that, in the linear BSDE \eqref{eq:LBSDE} with source term $Z^{BS}$, we have \red{$V_t = v^L(t,S_t)$, where $v^L$ solves the PDE \eqref{e:linear_PDE}. 
The classical Feynman-Kac representation of the solution to this PDE yields}
\[
V_{0}^{L} =  \mathbb{E}\left[e^{-rT}\Phi\left(S_{T}\right)+\int_{0}^{T}e^{-rs}\left(R \, C_{\alpha}|Z_{s}^{BS}|\sqrt{\left(s+\Delta\right)\wedge T-s}\right)\ds\right],
\]
which we can rewrite as
\begin{equation} \label{e:formula_V}
V_{0}^{L} = \mathbb{E}\left[e^{-rT}\Phi\left(S_{T}\right)+Te^{-rU} R \, C_{\alpha}|Z_{U}^{BS}|\sqrt{\left(U+\Delta\right)\wedge T-U} \right],
\end{equation}
where $U$ is an independent random variable uniformly distributed on $[0,T]$.
We introduce this additional randomization so to avoid any discretisation of the integral $\int_0^T \ds$ -- which would introduce a bias.  
\red{Taking into account that $Z_0^L$ is the derivative of $V_0^L$ with respect to the risky asset $S_0$ times the diffusion coefficient of $S$ computed at time 0, we apply again automatic differentiation formula (see \cite[Proposition 6.2.1]{nualart:06}) on \eqref{e:formula_V} to obtain}
\begin{equation} \label{e:formula_Z}
\begin{aligned}
Z_{0}^{L} = \mathbb{E}\bigg[e^{-rT}\Phi\left(S_{T}\right)\frac{W_T}{T} +  Te^{-rU}\left(R \, C_{\alpha}|Z_{U}^{BS}|\sqrt{\left(U+\Delta\right)\wedge T-U}\right)\frac{W_U}{U}\bigg].
\end{aligned}
\end{equation}

We solve the linear BSDE $V^{L}$ for different payoffs (call, put and butterfly option) with a Nested Monte-Carlo algorithm (Nested MC), and compare again with finite difference method as in the previous sections.
\red{The Nested MC algorithm consists in approximating the outer expectation in \eqref{e:formula_V} and in \eqref{e:formula_Z} with an empirical mean over $M$ i.i.d.\til samples, and the inner expectation \eqref{eq:formulaDelta} with an empirical mean over $N$ i.i.d.\til samples.
More precisely, the Nested MC estimator for $V_0^L$ is
\begin{multline*}
\widehat{V}_0^L
=
\frac 1 M \sum_{m = 1}^M\Bigg[
e^{-r \,T} \Phi \bigl(S_0 e^{(r- \sigma^2/2)T + \sigma \sqrt{T} X^m} \bigr)
\\
+
T  e^{-r \, T} 
R \, C_{\alpha}
\biggl|
\frac 1 N
\sum_{n=1}^N
\biggl(\Phi \Bigl(S_m e^{(r- \frac{\sigma^2}2)(T - U^m) + \sigma \sqrt{T - U^m} Y^{m,n}} \Bigr)
-
\Phi(S^m) \biggr)
\frac{Y^{m,n}}{\sqrt{T-U^m}}
\biggr|
\sqrt{\left(U^m+\Delta\right)\wedge T - U^m}\Bigg],
\end{multline*}
where $S^m = S_0 e^{(r- \sigma^2/2) U^m + \sigma \sqrt{U^m} X^m}$, and $(X^m)_{m =1, \dots, M}$, $(Y^{m,n})_{m =1, \dots, M}^{n =1, \dots, N}$ are independent standard normal random variables.
More details and an analysis of the bias and variance of analogous Nested MC estimators can be found in \cite{GordyJuneja} and \cite{GilesHajiAli}.
Note that the problem is eventually non-linear, due to the presence of the absolute value function applied to the inner expectation.
In our tests, we use $N = 100$ and $M=100000$.}
The results are presented in Table \ref{tab:bsde_lin}.
\red{In this one-dimensional example, we observe a good agreement between the Nested MC algorithm and the finite difference solution, which we use as a benchmark.
While the finite difference can typically hardly go beyond dimension $d=2$ or $d=3$, the Nested MC provides a competitive method to estimate $(V^L,Z^L)$ in higher dimensions.
Moreover, comparing the values in Table \ref{tab:bsde_lin} with the corresponding lines in Tables \ref{tab:callNL}-\ref{tab:putNL}-\ref{tab:butterflyNL}, one can observe that the non-linear price and delta $V^{NL}, Z^{NL}$ are well approximated by their linear counterparts $V^L, Z^L$, as predicted by Theorem \ref{theorem: approximation:price}. 
}

\begin{table}[!htb] 
\centering
\small
\begin{tabular}{c | c c c c }
\toprule
 & \multicolumn{2}{c}{FD} & \multicolumn{2}{c}{Nested MC}\\
  \cline{2-5} 
&&&95\% CI&95\% CI\\
& $V(0,S_0)$ & $\nabla V(0,S_0)$  &  $V_0^L(S_0)$ & $Z_0^L(S_0)/(\sigma S_0)$  \\
\toprule
 Call, $K = 17$ & $3.9843$ & $0.8072$ & $[3.9796,3.9856]$ & $[0.8059,0.8082]$  \\
 \hline
 Call, $K = 20$ & $2.1971$ & $0.5852$ & $[2.1931,2.1979]$ & $[0.5843,0.5862]$  \\
 \hline
 Call, $K = 23$ & $1.0953$ & $0.3653$ & $[1.0924,1.0958]$ & $[0.3644,0.3659]$  \\
 \hline
 Put, $K = 17$ & $0.6249$ & $-0.1981$ & $[0.6233,0.6249]$ & $[-0.1983,-0.1977]$  \\
 \hline
 Put, $K = 20$ & $1.7950$ & $-0.4209$ & $[1.7937,1.7966]$ & $[-0.4216,-0.4205]$  \\
 \hline
 Put, $K = 23$ & $3.6502$ & $-0.6398$ & $[3.6468,3.6511]$ & $[-0.6407,-0.6393]$  \\
 \hline 
  Butterfly, $K = 11$ & $0.0414$ & $-0.0181$ & $[0.0412,0.0415]$ & $[-0.0181,-0.0180]$  \\
 \hline  
  Butterfly, $K = 20$ & $0.3112$ & $0.0021$ & $[0.3112,0.3119]$ & $[0.0021,0.0022]$  \\
 \hline  
  Butterfly, $K = 29$ & $0.0689$ & $0.0206$ & $[0.0686,0.0690]$ & $[0.0206,0.0207]$  \\
 \hline  
\end{tabular}
\vspace{0.3cm}
\caption{Nested MC algorithm for the BSDE $(V_t^{L},Z_t^{L})$.}
\label{tab:bsde_lin}
\end{table}

\subsection{Basket options in higher dimensions}
In this section we solve the non-linear BSDE in high dimensions using SRMDP algorithm. In this setting, traditional full grid methods like finite difference are not able to tackle the problem for dimension greater than $3$.

We consider call option on a basket of $d$ assets where the asset process is modelled by multi-dimensional geometric Brownian motion with constant correlation $\rho_{ij}= \rho = 0.75$ for $i\neq j$ and constant volatility $ \sigma_0 = 0.25.$ The full-rank volatility matrix $\sigma$ in model \eqref{eq:ito market} is then given by 
$$ \sigma\sigma^\top = \Sigma \text{ where } \Sigma:= \left( \Sigma_{ij}\right)_{1 \leq i, j \leq d} \text{ with } \Sigma_{ij} = \sigma^2_0 \rho, i \neq j \text{ and }  \Sigma_{ii} = \sigma^2_0.$$
Then, $A_0 :=\left( \left((\sigma^1 S^1_0)^\top, \ldots, (\sigma^d S^d_0)^\top \right)^\top\right)^{-1}$  where $\sigma^i$ is the $i$th row of $\sigma.$
The payoff is given by $$\Phi(S_T^1,\ldots,S_T^d) = \left( \sum_i p^i S_T^i - K\right)^+ .$$
The option expiration is set to $T=1$ year and the interest rate $r=0.02$. We suppose that weights $p_i =\frac{1}{d}$ for all $i$. The strike price $K$ equals $20$ and the initial values of the assets $S_0 = (S_0^1,\ldots,S_0^d)$ are specified in Table \ref{tab:basket}. The rest of the model parameters are the same as earlier. In this table, we present prices and deltas for different basket options with several underlyings. In the first column, classical crude Monte Carlo values are shown (MC $R=0$, IM was not considered). In the second column SRMDP values are displayed taking into account IM. \red{In order to cope with the curse of dimensionality, here we used the \textbf{LP1} version of the SRMDP algorithm, with $500, 200, 50, 20$ hypercubes and $1000, 1500, 2000, 2500$ simulations per hypercube, for $d=2,3,4,5$ respectively, 5 time steps were considered. With respect to Table \ref{tab:basket}, in the absence of analytical solutions to compare with, we can only assert that Monte Carlo numerical results taking into account IM are close to those without IM, as seen before with the same Monte Carlo solvers executed over one dimensional problems. }
\begin{table}[!htb]
\centering
\small
\begin{tabular}{ c | c r c r }
\toprule
 & \multicolumn{2}{c}{MC ($R=0$)} & \multicolumn{2}{c}{SRMDP ($R=0.02$)}\\
  \cline{2-5} 
 & 95\% CI  & 95\% CI  & 95\% CI  & 95\% CI   \\
 $S_0$ & $V_0^{BS}(S_0)$ & $Z_0^{BS}(S_0)A_0$ &  $V_0^{NL}(S_0)$ & $Z_0^{NL}(S_0) A_0$  \\
\toprule
 $(18,20)$ & $[1.5102,1.5113]$ & $[-0.0685,-0.0682]$ & $[1.5015,1.5468]$ & $[-0.0772,-0.0649]$  \\
           &                   & $[0.6237,0.6245]$ &                   & $[0.6297,0.6556]$  \\
 \hline
 $(18,20,22)$ & $[2.0067,2.0081]$ & $[-0.4676,-0.4671]$ & $[1.9915,2.0447]$ & $[-0.4756,-0.4641]$  \\
              &                   & $[0.3813,0.3817]$ &                   & $[0.3873,0.4167]$  \\
              &                   & $[0.7435,0.7443]$ &                   & $[0.7725,0.7882]$  \\    
 \hline 
 $(16,18,20,22)$ & $[1.4470,1.4481]$ & $[-0.6589,-0.6582]$ & $[1.4677,1.5090]$ & $[-0.6689,-0.6182]$  \\
                 &                   & $[0.1062,0.1064]$ &                   & $[0.0962,0.1374]$  \\
                 &                   & $[0.4334,0.4338]$ &                   & $[0.4234,0.4628]$  \\
                 &                   & $[0.6093,0.6100]$ &                   & $[0.5943,0.6310]$  \\               
 \hline    
 $(16,18,20,22,24)$ & $[1.9672,1.9676]$ & $[-1.0467,-1.0455]$ & $[1.9928,2.0692]$ & $[-1.0767,-1.0242]$  \\
                    &                   & $[-0.0855,-0.0852]$ &                   & $[-0.1155,-0.0752]$  \\
                    &                   & $[0.3342,0.3347]$ &                   & $[0.3042,0.3467]$  \\
                    &                   & $[0.5655,0.5662]$ &                   & $[0.5355,0.5762]$  \\
                    &                   & $[0.7039,0.7047]$ &                   & $[0.6839,0.7167]$  \\                                    
 \hline   
\end{tabular}
\vspace{0.3cm}
\caption{Prices and deltas for the basket call option.}
\label{tab:basket}
\end{table}

\red{\appendix 
\section{Appendix. On the implied volatility smiles from call an put options with IM correction}

\noindent We analyse more in detail the behavior of the implied volatility smiles obtained from the non-linear call and put option prices in Figures \ref{fig:callNL} and \ref{fig:putNL}.
As already pointed out at the end of Section \ref{s:fin_diff}, theses prices do not satisfy put-call parity, which explains why the implied volatilities from calls and puts do not coincide.

Under the non-linear pricing rule $V^{NL}$, the price of a call option is
\[
\begin{aligned}
V_0^{NL} = C(T,K) &:= 
\esp \Bigl[
e^{-rT}
\left( S_0 e^{rT - \int_0^T d^{\mathrm{Call}}(t)dt + \sigma \sqrt T G - \frac 12 \sigma^2 T} - K \right)^+
\Bigr]
\\
&=
F^T(r,\dC) e^{-rT} 
\esp \left[
\left( e^{\sigma \sqrt T G - \frac 12 \sigma^2 T} - \frac{K}{F^T(r,\dC) } \right)^+
\right]\\
&:= 
F^T(r,\dC) e^{-rT} C_{\mathrm{BS}} \left(\frac K{F^T(r,\dC)}, \sigma \sqrt T \right)
\end{aligned}
\]
where $F^T(r,d) = S_0 e^{rT - \int_0^T d(t)\dt}$ is the usual forward price of $S$ with interest rate $r$ and dividend rate $d(t)$, and we denote $C_{\mathrm{BS}}(X,v)$ the normalized BS call price with moneyness $X$ and total implied volatility $v$ (so that $C_{\mathrm{BS}}(0,v) = 1$).
In our case,
\[
\begin{aligned}
&d^{\mathrm{Call}}(t) = - C_{\alpha} R \, \sigma \sqrt{(t+\Delta) \wedge T - t},
\\
&d^{\mathrm{Put}}(t) = C_{\alpha} R \, \sigma \sqrt{(t+\Delta) \wedge T - t}.
\end{aligned}
\]
How do we imply the implied volatility $\sigC$? Using the same reference Black-Scholes price that leads to a constant smile when $R=0$ (as in Figures \ref{fig:callNL}-\ref{fig:putNL}): that is, we impose
\begin{multline} \label{e:def_implVol}
C(T,K) =
F^T(r,\dC) e^{-rT} C_{\mathrm{BS}} \left(\frac K{F^T(r,\dC)}, \sigma \sqrt T \right)
\\
=
F^T(r,d=0) e^{-rT} C_{\mathrm{BS}} \left(\frac K{F^T(r,d=0)}, \sigC(K) \sqrt T \right)
=
\mathrm{RefCallPrice}(T,K)
\end{multline}
and analogously for $\sigP(\cdot)$. Equation \eqref{e:def_implVol} is of the form
\[
C_{\mathrm{BS}} \left(X, \sigC(X F_1) \sqrt T \right)
 = \frac{F_2}{F_1} C_{\mathrm{BS}} \left(X \frac{F_1}{F_2}, \sigma \sqrt T \right)
= \esp \left[\left(\frac{F_2}{F_1}e^{\sigma \sqrt T -\frac12 \sigma^2 T}- X\right)^+ \right],
\qquad X \ge 0,
\]
with $F_2 = F^T(r,\dC) \maj F_1 = F^T(r,0)$.
This already shows that: 
\begin{enumerate}
\item There is a smile: $\sigC(\cdot)$ is not constant.
\item $\sigC(K) \ge \sigma$ for all $K$, since $F_2 \ge F_1$.
\end{enumerate}
The same conclusions (1) and (2) also hold for $\sigP(\cdot)$, for which we have $\dP \ge0$, hence $F^T(r,\dP) \mino F^T(r,0)$.

Why are the implied volatilities of calls and puts with IM different?
Because the reference price satisfies the following parity relation
\[
\mathrm{RefCallPrice}(T,K) -
\mathrm{RefPutPrice}(T,K) 
= e^{-rT} \left( F^T(r,0) - K \right)
\]
while the call and put prices with IM satisfy 
\[
\begin{aligned}
C(T,K) - P(T,K)
&=
e^{-rT} \biggl[
F(\dC)  
C_{\mathrm{BS}}(\dC)  
- 
F(\dP)
P_{\mathrm{BS}}(\dP)  
\biggr]
\\
&= e^{-rT} \left(F(\dC) - K \right)
\\
&\quad
+ e^{-rT} \biggl[
F(\dC) \left(P_{\mathrm{BS}}(\dC) - P_{\mathrm{BS}}(\dP) \right)
\\
&\quad 
+ P_{\mathrm{BS}}(\dP) \left( F(\dC) - F(\dP) \right)
\biggr],
\end{aligned}
\]
where we have slightly simplified notations by setting $F(\dC) = F^T(r,\dC)$, $F(\dP) = F^T(r,\dP)$, and $C_{\mathrm{BS}}(d) = C_{\mathrm{BS}} \left(\frac K{F^T(r,d)}, \sigma \sqrt T \right)$, $P_{\mathrm{BS}}(d) = P_{\mathrm{BS}} \left(\frac K{F^T(r,d)}, \sigma \sqrt T \right)$.
Note that the right hand side is not an affine function of $K$.

We can now assess the sign of the ATMF skew $\partial_K \sigma(K)|_{K = F^T(r,0)}$.
Taking derivatives with respect to $K$ in \eqref{e:def_implVol}, we get
\[
\begin{aligned}
\Prob \left(S_T^{r, \dC, \sigma} \maj K\right)
&=
F^T(r,0) 
\left[
\partial_K C_{\mathrm{BS}} \biggl(\frac K{F^T(r,0)}, \sigC(K) \sqrt T \biggr)
+
\partial_\sigma C_{\mathrm{BS}} \biggl(\frac K{F^T(r,0)}, \sigC(K) \sqrt T \biggr)
\partial_K \sigC(K)
\right]
\\
&=
-\Prob \left(S_T^{r,d=0, \sigC(K)} \maj K\right)
+
K \sqrt{T} \phi
\biggl(
d_2
\biggl(
\frac K{F^T(r,0)}, \sigC(K) \sqrt T
\biggr)
\biggr)
\partial_K \sigC(K),
\end{aligned}
\]
or yet
\begin{equation} \label{e:implVolSkew}
\sqrt T \partial_K \sigC(K)
=
\frac
{-\Prob \left(S_T^{r, \dC, \sigma} \maj K\right) + \Prob \left(S_T^{r,d=0, \sigC(K)} \maj K\right)}
{K \phi
\biggl(
d_2
\biggl(
\frac K{F^T(r,0)}, \sigC(K) \sqrt T
\biggr)
\biggr)}
\end{equation}
where $\phi(\cdot)$ is the standard normal density.
Note that $\partial_K \sigC(K)$ has the sign of the numerator $\mathrm{Num^{CallSkew}}$ in \eqref{e:implVolSkew}.
Setting
\[
q(d,\sigma)
:= \Prob \left(S_T^{r, d, \sigma} \maj K\right)
= N \biggl(
d_2 \biggl(
\frac K{F^T(r,d)}, \sigma \sqrt T
\biggr)
\biggr)
= 
N \biggl(
\frac{\ln \frac{F^T(r,d)} K}{\sigma \sqrt T}
- \frac 12 \sigma \sqrt T
\biggr)
\]
we have
\[
\partial_\sigma q(d,\sigma)
= \sqrt T \phi(d_2) 
\biggl(
-\frac{\ln \frac{F^T(r,d)} K}{\sigma^2 T}
- \frac 12
\biggr),
\qquad
\partial_d q(d,\sigma)
= \frac 1{\sigma \sqrt T}
\phi(d_2) 
\frac{\partial_d F^T(r,d)}{F^T(r,d)}
=
- \frac T{\sigma \sqrt T}
\phi(d_2) 
\]
the last derivative being taken in the case of a constant dividend rate $d$ (which is our case up to time $t \le T-\Delta$).
We can write the numerator in \eqref{e:implVolSkew} as 
\[
\begin{aligned}
\mathrm{Num^{CallSkew}}
&=
q(d=0, \sigC(K)) - q(\dC, \sigma)
\\
&=
q(d=0, \sigC(K)) - q(d=0, \sigma)
+ q(d=0, \sigma) - q(\dC, \sigma)
=: 
\mathrm{Num_1^{CallSkew}}
+
\mathrm{Num_2^{CallSkew}}.
\end{aligned}
\]
We now consider the ATMF point, that is $K=F^T(r,0)$. 
Since, in this case,
\[
\partial_\sigma q(d=0,\sigma)
= - \frac 12 \sqrt T \phi(d_2) \mino 0,
\]
we have $\mathrm{Num_1^{CallSkew}} \mino 0$, recalling that $\sigC(K) \maj \sigma$.
Moreover, since $\partial_d q(d,\sigma) \mino 0$, we also have $\mathrm{Num_2^{CallSkew}} \mino 0$, recalling that $0 \maj \dC$.
Overall, we have obtained
\[
\sqrt T \partial_K \sigC(K)|_{K = F^T(r,0)}
\mino 0.
\]
}


\end{document}